\documentclass{lmcs}
\pdfoutput=1

\keywords{conservativity, coherence, strictification, dependent type theory, homotopy type theory, definitional equalities, strict equalities}

\usepackage[english]{babel}
\usepackage[utf8]{inputenc}
\usepackage[T1]{fontenc}

\usepackage{tikz}
\usetikzlibrary{cd}
\usetikzlibrary{decorations.pathmorphing}

\usepackage{xcolor}
\usepackage{multirow}
\usepackage{marginnote}

\usepackage{mathpartir}

\usepackage{stmaryrd}
\usepackage{amsmath}
\usepackage{amsthm}
\usepackage{amssymb}
\usepackage{mathtools}
\usepackage{environ}
\usepackage{bm}
\usepackage{bbm}
\usepackage{xspace}
\usepackage{hyperref}

\hypersetup{
  breaklinks=true,
  colorlinks=true
}

\usepackage[nameinlink]{cleveref}

\newcommand{\renewtheorem}[1]{%
  \expandafter\let\csname #1\endcsname\relax
  \expandafter\let\csname c@#1\endcsname\relax
  \expandafter\let\csname end#1\endcsname\relax
  \newtheorem{#1}%
}

\theoremstyle{plain}

\renewtheorem{thm}{Theorem}[section]
\renewtheorem{cor}[thm]{Corollary}
\renewtheorem{lem}[thm]{Lemma}
\renewtheorem{slem}[thm]{Sublemma}
\renewtheorem{prop}[thm]{Proposition}
\renewtheorem{asm}[thm]{Assumption}

\theoremstyle{definition}

\renewtheorem{rem}[thm]{Remark}
\renewtheorem{rems}[thm]{Remarks}
\renewtheorem{exa}[thm]{Example}
\renewtheorem{exas}[thm]{Examples}
\renewtheorem{defi}[thm]{Definition}
\renewtheorem{conv}[thm]{Convention}
\renewtheorem{conj}[thm]{Conjecture}
\renewtheorem{prob}[thm]{Problem}
\renewtheorem{oprob}[thm]{Open Problem}
\renewtheorem{oprobs}[thm]{Open Problems}
\renewtheorem{algo}[thm]{Algorithm}
\renewtheorem{obs}[thm]{Observation}
\renewtheorem{desc}[thm]{Description}
\renewtheorem{fact}[thm]{Fact}
\renewtheorem{qu}[thm]{Question}
\renewtheorem{oqu}[thm]{Open Question}
\renewtheorem{pty}[thm]{Property}
\renewtheorem{clm}[thm]{Claim}
\renewtheorem{nota}[thm]{Notation}
\renewtheorem{com}[thm]{Comment}
\renewtheorem{coms}[thm]{Comments}

\theoremstyle{plain}\newtheorem*{thm*}{Theorem}
\theoremstyle{plain}\newtheorem*{conj*}{Conjecture}
\theoremstyle{definition}\newtheorem{con}[thm]{Construction}
\theoremstyle{definition}\newtheorem*{exas*}{Examples}

\definecolor{RedOrange}{HTML}{F26035}

\input{macros.tex}

\newcommand{\reptype}{{\mathsf{type}_{\rep}}}
\newcommand{\repty}{{\mathsf{ty}_{\rep}}}

\newcommand{\GenTy}{\mathsf{GenTy}}
\newcommand{\GenRepTy}{\mathsf{GenRepTy}}

\newcommand{\lexO}{\Sigma,\Eq}
\newcommand{\repO}{\Sigma,\Pi_{\rep},\Eq}

\newcommand{\lexpre}{\Sigma}
\newcommand{\reppre}{\Sigma,\Pi_{\rep}}
\newcommand{\lccpre}{\Sigma,\Pi}
\newcommand{\lexinfty}{\Sigma,{\Id}}
\newcommand{\repinfty}{\Sigma,\Pi_{\rep},{\Id}}
\newcommand{\lccinfty}{\Sigma,\Pi,{\Id}}

\newcommand{\inner}[1]{{\color{RedOrange}{\mathbf{#1}}}}
\newcommand{\innerSym}[1]{{\color{RedOrange}{\bm{#1}}}}

\newcommand{\iS}{\inner{S}}

\newcommand{\iTy}{\inner{Ty}}
\newcommand{\iTm}{\inner{Tm}}

\newcommand{\icong}{\mathrel{\innerSym{\cong}}}
\newcommand{\iid}{\inner{iid}}

\newcommand{\iId}{\inner{Id}}
\newcommand{\irefl}{\inner{refl}}
\newcommand{\iJ}{\inner{J}}
\newcommand{\iJb}{\inner{J_{\mathnormal{\beta}}}}

\newcommand{\iOb}{\inner{Ob}}
\newcommand{\iHom}{\inner{Hom}}
\newcommand{\iEqHom}{\inner{EqHom}}
\renewcommand{\iid}{\inner{id}}
\newcommand{\icomp}{\inner{comp}}
\newcommand{\icirc}{\mathbin{\innerSym{\circ}}}

\newcommand{\iotimes}{\mathbin{\innerSym{\otimes}}}
\newcommand{\iI}{\inner{I}}

\newcommand{\Th}{\CT}
\newcommand{\CMonCat}{\mathbf{MonCat}}
\newcommand{\CStrMonCat}{\mathbf{StrMonCat}}

\newcommand{\CTele}{\mathbf{Tele}}

\newcommand{\Free}{\mathsf{Free}}
\newcommand{\CSetoid}{\mathbf{Setoid}}
\newcommand{\fib}{\mathsf{fib}}
\newcommand{\CFam}{\mathbf{Fam}}

\newcommand{\idToHpty}{{\mathsf{id}\text{-}\mathsf{to}\text{-}\mathsf{hpty}}}
\newcommand{\Ehat}{{\widehat{E}}}
\newcommand{\ECat}{{E\text{-}\CCat}}
\newcommand{\hrefl}{\mathsf{hrefl}}
\newcommand{\El}{\mathsf{El}}

\author{Rafaël Bocquet}
\address{Department of Programming Languages and Compilers, Eötvös Loránd University, Budapest, Hungary}
\email{bocquet@inf.elte.hu}
\urladdr{\url{https://rafaelbocquet.gitlab.io/}}
\thanks{The author was supported by the European Union, co-financed by the European Social Fund (EFOP-3.6.3-VEKOP-16-2017-00002).}

\title{Towards coherence theorems for equational extensions of type theories}
\date{\today}

\begin{document}

\begin{abstract}
  We study the conservativity of extensions by additional strict equalities of dependent type theories (and more general second-order generalized algebraic theories).
  The conservativity of Extensional Type Theory over Intensional Type Theory was proven by Hofmann.
  Our goal is to generalize such results to type theories without the Uniqueness of Identity Proofs principles, such as variants of Homotopy Type Theory.

  For this purpose, we construct what is essentially the $\infty$-congruence on the base theory that is freely generated by the considered equations.
  This induces a factorization of any equational extension, whose two factors can be studied independently.
  We conjecture that the first factor is always an equivalence when the base theory is well-behaved.
  We prove that the second factor is an equivalence when the $\infty$-congruence is $0$-truncated.
\end{abstract}

\maketitle

\section{Introduction}

Equality and computation are central components of type theories.
The computational content of a type theory is presented by computation rules (often called $\beta$-rules, or $\iota$-rules for inductive types), and perhaps uniqueness rules (usually called $\eta$-rules) or more exotic rules (such as the $\nu$-rules considered in \cite{NewEqsNeutrals}).
This computational content is typically explained by the means of a normalization algorithm.
In presence of identity types, there is a distinction between two notions of equality or identification between terms of a type theory.
Internally, the identity types provide the notion of \emph{internal identification}, also called \emph{propositional equality} or \emph{typal equality}.
Externally, we can also compare terms up to \emph{strict equality}, which is the proof-irrelevant equality of our metatheory.
Strict equality is also often called \emph{definitional} or \emph{judgemental} equality.

When working internally to a type theory, it is desirable to have as many strict equalities as possible.
Indeed, strict equality can be implicitly and silently coerced over.
On the other hand, internal identifications require explicit transports and coercions, which quickly clutter the terms of the theory.
Conversely, the trade-off is that type theories with additional strict equalities have fewer models, and their semantics are therefore more complicated.

Hofmann proved in \cite{HofmannCons} a conservativity theorem, showing that all internal identifications in a type theory can conservatively be turned into strict equalities, when the base type theory satisfies some principles.
The most important of these principles is the Uniqueness of Identity Proofs (UIP) principle, which states that there is an identification between any two elements of identity types.
A more syntactic proof was later given by Oury~\cite{OuryConservativity} for the calculus of constructions.
Oury's proof had some issues, mainly due to a presentation of the syntax of type theory with too few annotations.
An improvement of Oury's proof and a presentation of this result as a constructive and effective syntactic translation has been given recently by Winterhalter\etal~\cite{ElimRefl}.

Since Hofmann's proof of conservativity, there has been a lot of interest going into the study of type theories with non-trivial higher dimensional content inconsistent with UIP~\cite{GroupoidModel}, and their semantics in homotopy theoretic~\cite{HomotopyTheoreticModels, SimplicialModel} and $\infty$-categorical structures~\cite{LangLexInftyCats}.
For type theories without UIP, strict equalities are even more important, because they are automatically coherent.
Thus having more strict equalities means that we escape not only ``transport hell'', but also ``higher-dimensional transport and coherence hell''.
Conversely, it is in practice much harder to justify strict equalities in many homotopy theoretic models.
Some authors have even considered weak variants of the basic computation rules of type theories.
For instance, weak identity types, whose computation rule only holds up to internal paths, have been introduced by \cite{PropIdTypes}, under the name of propositional identity types.
The path types of cubical type theories~\cite{CTT} also only satisfy weakly the computation rule of identity types.
Other type structures can be weakened similarly, and we can even consider type theories whose computation rules are all expressed by internal identifications instead of strict equalities.
At the level of types and universes, instead of assuming that each type former is strictly classified by some code in the universe, we can ask for them to be classified only up to type equivalence.
These weak Tarski universes have been introduced in \cite{WeakTarski}.
The fact that homotopy type theory with strict univalent universes, rather than weak universes, can be interpreted in every $(\infty,1)$-topos has only been established recently~\cite{InftyToposesUniverses}.

In this setting, we can wonder how type theories with varying amounts of strict equalities can be compared.
More precisely, we wish to know how to establish coherence and strictification theorems, that would allow us, when working internally to a model of a weak type theory, to pretend that it satisfies more strict equalities than it actually does, by replacing it by an equivalent stricter model.
The question of the conservativity of strict identity types over weak identity types has been asked at the TYPES 2017 conference~\cite{WeakJTypes}, motivated by the fact that the path types in cubical type theory only satisfy the elimination principle of weak identity types.
This was also the original motivation for the present paper.

We give some examples of weakenings and extensions of homotopy type theory that ought to be equivalent to standard homotopy type theory.
\begin{exas}\label{exas:applications} \hfill
  \begin{itemize}
    \item Weakening the $\beta$ and $\eta$ computation rules of identity types, $\Sigma$-types, inductive types, etc, gives a weaker variant of HoTT.
    \item We can add strict equalities that make the addition on natural numbers into a strictly associative and commutative operation.
      That is, while the inductive definition of $(- + -) : \Nat \to \Nat \to \Nat$ only satisfies the strict equalities $0 + y = y$ and $(\mathsf{suc}\ x) + y = \mathsf{suc}\ (x + y)$, we would add the strict equalities $x + 0 = 0$, $x + (\mathsf{suc}\ y) = \mathsf{suc}\ (x + y)$, $x + y = y + x$, $(x + y) + z = x + (y + z)$, etc.
    \item Similarly, we could make the composition of paths into a strictly associative operation, optionally with strict inverses.
    \item We can extend the theory with a universe of strict proposition~\cite{SProp} that is equivalent to the universe of homotopy propositions.
    \item Similarly, we can extend the theory with universes of definitional categories, definitional rings, etc, that satisfy strictly the equations of the theories of categories, rings, etc.
    \item We can extend the theory with a universe of ``definitionally'' pointed types $\mathsf{DPtType}$, equivalent to the universe of pointed types $\mathsf{PtType} \triangleq (A : \UU) \times A$, with a smash product operation $(- \wedge -) : \mathsf{DPtType} \to \mathsf{DPtType} \to \mathsf{DPtType}$ with more strict equalities than the smash product of $\mathsf{PtType}$.
      This would provide an alternative interpretation of Brunerie's rewriting based method to prove that the smash product is a symmetric monoidal product on pointed types~\cite{BrunerieSmash}.
  \end{itemize}
\end{exas}

Some progress has been made by Isaev in \cite{IsaevMorita}.
In that paper, Isaev defines the notion of Morita equivalence between type theories, and gives some characterizations of that notion.
A first conservativity result in the absence of UIP is also proven, showing that type theories with weak or strict unit types are Morita equivalent.

The constructions by Isaev~\cite{IsaevMS} and Kapulkin and Lumsdaine~\cite{HoThTT,HoInvCwA}, of Quillen model or semi-model structures over the categories of models of type theories, are also extremely relevant for our work.
In particular, as remarked in \cite{HoThTT}, Hofmann's conservativity theorem proves exactly that the morphism $\Init_{\mathsf{ITT}} \to \Init_{\mathsf{ETT}}$ between the initial models of intensional type theory and extensional type theory is a trivial fibration of their semi-model structure.
The weak equivalences of the same semi-model structure correspond to a weaker notion of conservativity than trivial fibrations.
Isaev's definition of Morita equivalence relies on that notion of weak equivalence.

This paper builds on top of the aforementioned work.
While Isaev considers the notion of Morita equivalence for arbitrary morphisms between type theories, we restrict our attention to the equational extensions $\Th \to \Th_E$ of a weak type theory to a strict type theory $\Th_{E}$, by a family of equations $E$, which should hold weakly in $\Th$ and strictly in $\Th_E$.
We then establish sufficient conditions for the theories $\Th$ and $\Th_{E}$ to be Morita equivalent.

The situation can be compared to other well-known coherence theorems, such as Mac Lane's coherence theorem for monoidal categories~\cite{MacLaneCoh}.
They can often be stated in multiple different ways.
For example, here are two related ways to state the coherence theorem for monoidal categories.
\begin{enumerate}
  \item\label{itm:coh_moncat_1} Every (weak) monoidal category is monoidally equivalent to a strict monoidal category.
  \item\label{itm:coh_moncat_2} In a freely generated monoidal category, every diagram made up of associators and unitors commutes.
\end{enumerate}

The statement (\ref{itm:coh_moncat_1}) is generally the one that we want to use: it allows us to work with any weak monoidal category as if it was strict.
The statement (\ref{itm:coh_moncat_2}) is however perhaps easier to prove, because free monoidal categories can be seen as syntactic objects, that are relatively easy to describe explicitly and understand.
See \cite{joyal1991geometry} for a proof of the statement (\ref{itm:coh_moncat_1}) that relies on the statement (\ref{itm:coh_moncat_2}).
In the case of monoidal category, it is actually possible to prove the statement (\ref{itm:coh_moncat_1}) more directly using representation theorems similar to the Yoneda lemma.
This kind of approach does not seem suitable for the coherence theorems that we are interested in.

The main result of this paper is a coherence theorem for type theories that is analogous to the fact the statement (\ref{itm:coh_moncat_1}) can be deduced from the statement (\ref{itm:coh_moncat_2}).
It states that to establish the conservativity of the equational extension $\Th \to \Th_E$, it suffices to check that the $\infty$-congruence over $\Th$ freely generated by the equations of $E$ exists and is $0$-truncated.
This $0$-truncatedness condition encodes the same idea as the fact that every diagram made up of associators and unitors commutes in a freely generated monoidal category.

In general, we work at the level of \emph{second-order generalized algebraic theories} (SOGATs) and their classifying $(\reppre)$-CwFs.
SOGATs correspond to the class of theories classified by representable map categories, which were introduced by Uemura~\cite{GeneralFrameworkTT}.
First-order generalized algebraic theories are also SOGATs, so our results also apply to theories that are not type theories, such as monoidal categories, \etc
A SOGAT $\Th$ is specified by its classifying $(\reppre)$-CwF, also denoted by $\Th$, which can be seen as a syntactic model of some type theory.
A SOGAT has functorial semantics in $(\reppre)$-CwFs.
Even if one is only interested in a specific theory, working with SOGATs explicitly is advantageous.
Indeed, many semantic conditions can instead be stated more concisely and syntactically directly at the level of the $(\reppre)$-CwF $\Th$.


The notion of $\infty$-congruence over a SOGAT $\Th$ is specified at the level of the classifying $(\reppre)$-CwF.
Informally, a $1$-congruence over $\Th$ is an extension of the set-valued $(\reppre)$-CwF $\Th$ to a setoid-valued $(\reppre)$-CwF; this equips the sets of contexts, substitutions, types and terms of $\Th$ with equivalence relations satisfying some conditions.
Instead, an $\infty$-congruence should be an extension of $\Th$ to a space-valued $(\reppre)$-CwF; this equips all components of $\Th$ with $\infty$-groupoid structures.

Defining $\infty$-congruences require choosing a model of $\infty$-groupoids among many.
Instead of working with \eg Kan simplicial sets, we rely on the fact that $\infty$-groupoids are closely related to models of type theory with identity types.
Instead of working with space valued $(\reppre)$-CwFs, with work with $(\reppre)$-CwFs with identity types, \ie $(\repinfty)$-CwFs.
In the original presentation of this work, this was inspired by Brunerie's type-theoretic definition of $\infty$-groupoid~\cite{BrunerieThesis}.
Another point of view is that models of type theory with identity types correspond to structured $\infty$-categories.
When working with $(\repinfty)$-CwFs, we morally work with representable map $\infty$-categories.
Representable map $\infty$-categories have recently been used by Nguyen and Uemura~\cite{InftyTypeTheories} to give functorial semantics to $\infty$-type theories, generalizing the functorial semantics of $1$-type theories in representable map categories.

We don't give any application of our coherence theorem in this paper, which is instead focused on proving general results that hold for any SOGAT.
Two elements are still missing for applications, discussed in more details in~\cref{sec:discussion}.
First, we need a proof of~\cref{conj:embedding_from_ext_univ}, which essentially states that our construction of the $\infty$-congruence generated by $E$ is meaningful when $\Th$ satisfies external univalence~\cite{ExternalUnivalence}.
Second, we need ways to prove the $0$-truncatedness conditions for concrete theories.
We believe that in most cases, this $0$-truncatedness can be obtained as a consequence of a homotopy normalization proof.
A strict normalization states that every term has a unique normal form; homotopy normalization states instead that every term has a contractible space of normal forms.
Uemura has recently given a general proof scheme for normalization of $\infty$-type theories~\cite{UemuraNormalization}, which can be used to prove such homotopy normalization results.
Because we work with $(\repinfty)$-CwFs rather than $\infty$-type theories, the proof does not directly carry over to our setting, but we believe that the results we need can be proven using similar methods.

\subsection*{Outline of the paper}

In~\cref{sec:background}, we introduce our notations and conventions, and review the notion of Category with Families (CwFs) and associated definitions.
We make use of the internal language of presheaf categories to describe our constructions.
In particular, instead of working with CwFs externally, we often prefer to work with families internally to presheaf categories.

In~\cref{sec:weak_id_types}, we define the structures of weak identity types, and show that these structures can be lifted from a family to its family of telescopes.
This means that these structures are well-behaved even in the absence of $\Sigma$-types.

In~\cref{sec:sogats}, we recall the definition of second-order generalized algebraic theories (SOGATs) and their functorial semantics in $(\reppre)$-CwFs.
We follow the presentation of the author's paper on external univalence for SOGATs~\cite{ExternalUnivalence}.
Examples include the theories of categories, monoidal categories, \etc, as well as type theory with various type structures.
We also recall the notion of external univalence for SOGATs equipped with homotopy relations.
We also define classes of trivial fibrations, fibrations and weak equivalences between models of SOGATs equipped with homotopy relations.

In~\cref{sec:equational_extensions}, we introduce the notion of equational extension $\Th \to \Th_E$ of a SOGAT, where $E$ is a collection of homotopies in $\Th$ that become strict equalities in $\Th_E$.
We discuss some categorical and type-theoretic examples.
We also specify the notion of conservativity an equational extension, by adapting Isaev's definition of Morita equivalence~\cite{IsaevMorita}.

In~\cref{sec:partial_saturation}, we introduce the notion of partial saturation: a $(\repinfty)$-CwF $\CC$ with a morphism $\Th \to \CC$ is partially saturated with respect to a collection of homotopies $E$ when the homotopies in $E$ can be lifted to the identity types of $\CC$.
This generalizes the notion of (full) saturation used in the specification of external univalence.
We then define $\Th_{\Ehat}^\infty$ as the initial object among $(\repinfty)$-CwFs under $\Th$ that are partially saturated with respect to $E$.
It can be viewed as the $\infty$-congruence on $\Th$ that is freely generated by the homotopies of $E$; the elements of its identity types are compositions and higher-dimensional compositions of the lifts of the homotopies in $E$.
Our next step is to consider the following diagram.
\[ \begin{tikzcd}
    \Th
    \ar[rr]
    \ar[dd]
    \ar[rd]
    && \Th_{E}
    \ar[dd, two heads, "{\sim}"]
    \\
    &
    \Th_{\Ehat}
    \ar[ld, two heads, "{\sim}"]
    \ar[ru]
    &
    \\
    \Th^\infty_{\Ehat}
    \ar[rr]
    && \Th^1_{E} \rlap{\ .}
  \end{tikzcd} \]
The model $\Th_{\Ehat}$ is the subCwF of $\Th_{\Ehat}^{\infty}$ spanned by the contexts and types that do not contain identity types.
The CwF $\Th_E^1$ is obtained by adding equality reflection to $\Th_\Ehat^\infty$, or equivalently by adding extensional equality types to $\Th_E$.

To prove that the equational extension $\Th \to \Th_E$ is a Morita equivalence, it then suffices to show that $\Th \to \Th_{\Ehat}$ and $\Th_{\Ehat} \to \Th_E$ are both equivalences.

We view proving that $\Th \to \Th_{\Ehat}$ is an equivalence (or equivalently $\Th \to \Th_{\Ehat}^\infty$, by $2$-out-of-$3$) as a proof that $\Th_{\Ehat}^{\infty}$ is really an $\infty$-congruence.
\begin{conj*}[Simplified statement of~\cref{conj:embedding_from_ext_univ}]
  If $\Th$ satisfies external univalence and is cofibrant, then $\Th \to \Th^\infty_{\Ehat}$ is an equivalence.
  \defiEnd{}
\end{conj*}
This conjecture will be discussed in~\cref{sec:discussion}.
For most examples, external univalence can be proven using the tools of~\cite{ExternalUnivalence}.

The main theorem of the paper involves the second morphism $\Th_{\Ehat} \to \Th_E$.
\begin{thm*}[Simplified statement of~\cref{thm:main_theorem}]
  If $\Th^\infty_{\Ehat}$ is merely $0$-truncated relatively to $\Th \to \Th^\infty_{\Ehat}$, then $\Th_{\Ehat} \to \Th_E$ is an equivalence.
\end{thm*}
The $(\repinfty)$-CwF $\Th^\infty_{\Ehat}$ being $0$-truncated means that all of its types are h-sets.
Being merely $0$-truncated means that we only know the mere existence of the terms witnessing $0$-truncatedness.
Being $0$-truncated relatively to $\Th \to \Th^\infty_{\Ehat}$ means that this is only valid over contexts in the image of $\Th \to \Th^\infty_{\Ehat}$.

In~\cref{sec:fibrant_congruences} we discuss congruences and quotients for models of first-order generalized algebraic theories.
The proof of our main theorem requires the construction of quotients of $(\reppre)$-CwFs by a congruence; a congruence is an extension of a set-valued model to a setoid-valued model.
The categories of models of GATs are cocomplete, so quotients always exists, but they are ill-behaved in general: the quotient inclusions are not always trivial fibrations (surjective on every sort).
We show that this can be fixed by restricting our attention to \emph{fibrant} congruences; a fibrant congruence is an extension of a set-valued model of a setoid-valued model satisfying some fibration condition.
The quotient inclusion of a fibrant congruence is always a trivial fibration, and conversely every trivial fibration is the quotient inclusion of a fibrant congruence.

In~\cref{sec:retracts_quotients} we prove~\cref{thm:main_theorem} by showing that under its assumptions, the map $\Th_{\Ehat} \to \Th_E$ is a retract of some quotient inclusion.
We rely on the theory of fibrant congruences developed in the previous section.

In~\cref{sec:discussion} we discuss some conjectures and the future work that would be required for applications.


\section{Background}\label{sec:background}

We recall in this section the semantics of type theories in categories with families (CwFs), and introduce the tools and notations that we will use in this paper.
We will make use in particular of the internal type-theoretic language of the presheaf category $\CPsh(\CC)$ as a tool to define and work with type-theoretic structures over a CwF $\CC$.

\subsection{Metatheory and basic notations}

We work in a constructive metatheory, say extensional type theory with enough universes and all QIITs.
In some cases we will need distinguish split and non-split surjective functions.
We adopt terminology from the HoTT book; we use the adverb \emph{merely} to indicate that we have an element of the propositional truncation of some set.
For example, a function $f : A \to B$ is a split surjection when for every $b : B$ there is some element $a : A$ such that $f(a) = b$.
A function $f : A \to B$ is a non-split surjection when f or every $b : B$ there merely exists some $a : A$ such that $f(a) = b$.

The sets of dependent function are written $(a : A) \to B(a)$, and dependent functions are introduced by $(a : A) \mapsto b(a)$.
We sometimes omit function arguments when they can be inferred from the context.
For instance, given a function $f : (a : A) \to B(a) \to C(a)$, and elements $a : A$ and $b : B(a)$, we may just write $f(b)$ or $f_{a}(b)$ instead of $f(a,b)$.

The sets of dependent pairs are written $(a : A) \times B(a)$, and dependent pairs are introduced by $(a,b)$.

We make use of the interpretation of extensional type theory into presheaf toposes.
Because our metatheory is constructive, any of our external constructions can also be interpreted in presheaf toposes.

\subsection{Families}

\begin{defi}
  A \defemph{family} $\MC$ consists of a set $\MC.\Ty : \SSet$ of \emph{types} and a dependent set $\MC.\Tm : \MC.\Ty \to \SSet$ of \emph{terms} or \emph{elements}.
  We often let $\MC$ stand for $\MC.\Ty$, and use $\MC.\Tm$ as an implicit coercion from types to sets.
  \defiEnd{}
\end{defi}

There is a category $\CFam$ of families and family morphisms. It is equivalent to the arrow category $\CSet^{\to}$.

\begin{defi}
  A \defemph{restriction} of a family $\MC$ is a map $\iota : \Ty' \to \MC.\Ty$.
  It induces a \defemph{restricted family} $\MC'$, with $\MC'.\Tm(A) = \MC.\Tm(\iota(A))$.
  There is an induced family morphism $\MC' \to \MC$.
  We say that a family morphism is \defemph{cartesian} if it induced by a family restriction, or equivalently if its action on terms is bijective.
  \defiEnd{}
\end{defi}

\subsection{Type-structures}

\begin{defi}
  A \defemph{$\Unit$-type structure} on a family $\MC$ consists of a type $\Unit : \MC$ and an isomorphism $\MC.\Tm(\Unit) \cong \{\star\}$.
  \defiEnd{}
\end{defi}

\begin{defi}
  A \defemph{$\Sigma$-type structure} on a family $\MC$ consists, for every $A : \MC$ and $B : A \to \MC$, of a type $\Sigma(A,B) : \MC$ and an isomorphism $\MC.\Tm(\Sigma(A,B)) \cong (a:A)\times B(a)$.
  \defiEnd{}
\end{defi}

\begin{defi}
  A \defemph{$\Pi$-type structure} on a family $\MC$ consists, for every $A : \MC$ and $B : A \to \MC$, of a type $\Pi(A,B) : \MC$ and an isomorphism $\MC.\Tm(\Pi(A,B)) \cong ((a:A)\to B(a))$.
  \defiEnd{}
\end{defi}

It is often the case that a family is not closed under arbitrary $\Pi$-types, but only under $\Pi$-types whose domains lie in another family.
\begin{defi}
  Let $\MC$ and $\MD$ be two families.
  A \defemph{$\Pi$-type structure} on $\MD$ with \defemph{arities} in $\MC$ consist, for every $A : \MC$ and $B : A \to \MD$, of a type $\Pi(A,B) : \MD$ and an isomorphism $\MD.\Tm(\Pi(A,B)) \cong ((a:A)\to B(a))$.
  \defiEnd{}
\end{defi}

\begin{defi}
  Let $\MC$ be a family, along with a restriction $\MC_\rep \to \MC$.
  A \defemph{$\Pi_\rep$-type structure} on $\MC$ consists of a $\Pi$-type structure on $\MC$ with arities in $\MC_\rep$.
  \defiEnd{}
\end{defi}

The types of $\MC_\rep$ are called representable types.
The restriction $\MC_\rep \to \MC$ is not necessarily injective, but we can often pretend that it is injective and omit it.

\begin{defi}
  A $(\reppre)$-family is a family $\MC$, along with a restriction $\MC_\rep \to \MC$, such that both $\MC$ and $\MC_\rep$ have $\Unit$- and $\Sigma$- type structures, and such that $\MC$ has a $\Pi_\rep$-type structure.
  \defiEnd{}
\end{defi}

Similarly, we may talk of $\lexpre$-families or $(\lccpre)$-families.

When working with $(\reppre)$-families, we may sometimes omit to mention $\MC_\rep$ in some constructions.
It should then be understood that the same construction is applied to $\MC_\rep$.
For example, the $\Unit$- and $\Sigma$- types of $\MC$ and $\MC_\rep$ are distinct, but in some constructions we may only refer to the $\Unit$- and $\Sigma$- types of $\MC$.

\subsection{Categories with families}

We review the notion of category with families (CwF), which is used to model type theory.

We first recall the definition of locally representable presheaf, which is used to encode the context extensions in models of type theory.
This definition is a reformulation of the notion of representable natural transformation, which has been used by Awodey~\cite{NaturalModels} to define natural models, which are equivalent to categories with families.
\begin{defi}
  Let $Y$ be a dependent presheaf over a presheaf $X$.
  We say that $Y$ is \defemph{locally representable} if for every element $x : X(\Gamma)$, the presheaf $Y_{\mid x} : (\CC / \Gamma)^\op \to \CSet$ is representable, where
  \begin{alignat*}{1}
    & Y_{\mid x}(\rho : \Delta \to \Gamma) \triangleq Y(x[\rho]).
  \end{alignat*}

  In that case, its representing object consists of an \defemph{extended context} $\Gamma.Y[x]$ and a \defemph{projection map} $\bm{p}_x : \Gamma.Y[x] \to \Gamma$.
  There is a \defemph{generic element} $\bm{q}_x : Y(x[\bm{p}_x])$.
  These satisfy the following universal property: for every other object $\Delta$, map $\gamma : \Delta \to \Gamma$ and element $a : Y(x[\gamma])$, there is a unique map $\angles{\gamma,a} : \Delta \to \Gamma.Y[x]$ such that $\bm{p}_x \circ \angles{\gamma,a} = \gamma$ and $\bm{q}_x[\angles{\gamma,a}] = a$.
  \defiEnd{}
\end{defi}

\begin{defi}
  A \defemph{category with families} (CwF) consists of a category $\CC$, with a terminal object, along with a global family $\MC$ in $\CPsh(\CC)$ such that $\MC.\Tm$ is locally representable.
  \defiEnd{}
\end{defi}

A $(\reppre)$-CwF is a CwF $\CC$ such that the family $\MC$ is a $(\reppre)$-family in $\CPsh(\CC)$.
Similarly, we have $\lexpre$-CwFs, $(\lccpre)$-CwFs, $(\lexO)$-CwFs, $(\repinfty)$-CwFs, \etc (Identity types and equality types will be specified in~\cref{sec:weak_id_types}).
The $1$-category of $(\reppre)$-CwFs is denoted by $\CCwf_{\reppre}$, and other $1$-categories of structured CwFs are written similarly.

Given a CwF $\CC$, we will write $(\Gamma \vdash A\ \type) \in \CC$ to indicate that $A$ is a type is context $\Gamma$, and $(\Gamma \vdash a : A) \in \CC$ to indicate that $a$ is a term of type $A$ in context $\Gamma$.
Occasionally we may also write $(\gamma:\Gamma \vdash A(\gamma)\ \type) \in \CC$ and $(\gamma:\Gamma \vdash a(\gamma) : A(\gamma)) \in \CC$.
We write $(\Gamma \vdash A\ \reptype) \in \CC$ if $A$ is a representable type over $\Gamma$, when $\CC$ is a $(\reppre)$-CwF.
We may write $(A\ \type) \in \CC$ or $(a : A) \in \CC$ when working over the empty context $1_\CC$.

We have freely generated $(\reppre)$-CwFs
\begin{mathpar}
  \Free_{\reppre}(\bm{\Gamma} \vdash),
  
  \Free_{\reppre}(\bm{\Gamma} \vdash \bm{A}\ \type),
  
  \Free_{\reppre}(\bm{\Gamma} \vdash \bm{A}\ \reptype),
  \\
  \Free_{\reppre}(\bm{\Gamma} \vdash \bm{a} : \bm{A}),
\end{mathpar}
where we use bold letters to indicate the generators.

We have maps
\begin{alignat*}{1}
  & I^\ty : \Free_{\reppre}(\bm{\Gamma} \vdash) \to \Free_{\reppre}(\bm{\Gamma} \vdash \bm{A}\ \type), \\
  & I^{\repty} : \Free_{\reppre}(\bm{\Gamma} \vdash) \to \Free_{\reppre}(\bm{\Gamma} \vdash \bm{A}\ \reptype), \\
  & I^\tm : \Free_{\reppre}(\bm{\Gamma} \vdash \bm{A}\ \type) \to \Free_{\reppre}(\bm{\Gamma} \vdash \bm{a} : \bm{A}), 
\end{alignat*}
which are generic extensions of a $(\reppre)$-CwF by a new type, representable type or term.

We also have maps
\begin{alignat*}{1}
  & E^\ty : \Free_{\reppre}(\bm{\Gamma} \vdash \bm{A},\bm{B}\ \type) \to \Free_{\reppre}(\bm{\Gamma} \vdash \bm{A} = \bm{B}), \\
  & E^{\repty} : \Free_{\reppre}(\bm{\Gamma} \vdash \bm{A},\bm{B}\ \reptype) \to \Free_{\reppre}(\bm{\Gamma} \vdash \bm{A} = \bm{B}), \\
  & E^\tm : \Free_{\reppre}(\bm{\Gamma} \vdash \bm{x},\bm{y} : \bm{A}) \to \Free_{\reppre}(\bm{\Gamma} \vdash \bm{x} = \bm{y}), 
\end{alignat*}
which are generic extensions of a $(\reppre)$-CwF by a new equation between two types, representable types or terms.
These maps are the codiagonals of $\{I^\ty,I^{\repty},I^\tm\}$.

We will consider various weak factorization systems and orthogonal factorization systems generated by subsets of $\{I^\ty,I^{\repty},I^\tm,E^\ty,E^{\repty},E^\tm\}$.
We may also similarly consider factorization systems on other categories of CwFs, \eg $\CCwf_{\lexpre}$ or just $\CCwf$.

In~\cref{sec:sogats} we will also consider similar factorization systems on the categories of models of arbitrary SOGATs.

We start with the (cofibrations, trivial fibrations) weak factorization systems.
\begin{defi}\label{def:tfib_cwf}
  A map $F : \CC \to \CD$ in $(\reppre)$-CwF is a (split) \defemph{trivial fibration} if its actions on types, representable types and terms are split surjective.
  \defiEnd{}
\end{defi}

\begin{defi}
  A map $F : \CC \to \CD$ in $(\reppre)$-CwF is a (non-split) \defemph{trivial fibration} if its actions on types, representable types and terms are merely surjective.
  \defiEnd{}
\end{defi}

The split trivial fibrations are the right class of the weak factorization system generated by $\{I^\ty,I^{\repty},I^\tm\}$.
Maps in the left class are called \emph{cofibrations}.

The conservativity of Extensional Type Theory over Intensional Type Theory can be stated using the notion of trivial fibration:
the unique morphism $\Init_{\mathsf{ITT}} \to \Init_{\mathsf{ETT}}$ between the initial models of the two type theories is a non-split trivial fibration in $\CCwf$, \ie its actions on types and terms are surjective.

The global picture of Hofmann's proof of conservativity~\cite{HofmannCons} is the following.
We construct a quotient $Q$ of the syntactic model $\Init_{\mathsf{ITT}}$.
By the universal properties of the quotient $Q$ and of the model $\Init_{\mathsf{ETT}}$, we obtain a section-retraction pair $(\Init_{\mathsf{ETT}} \to Q, Q \to \Init_{\mathsf{ETT}})$.
This exhibits the map $\Init_{\mathsf{ITT}} \to \Init_{\mathsf{ETT}}$ as a retract of the quotient inclusion $\Init_{\mathsf{ITT}} \to Q$.
Since the quotient inclusion is a non-split trivial fibration and non-split trivial fibration are closed under retracts, we obtain the desired conservativity result.
Most of the complexity of the proof lies in the construction of the quotient $Q$.

It is not possible to prove constructively that $\Init_{\mathsf{ITT}} \to \Init_{\mathsf{ETT}}$ is a split trivial fibration.
Indeed, if it were possible, it would be possible to decide equality of terms of $\Init_{\mathsf{ETT}}$ by deciding the equality of their lifts in $\mathsf{ITT}$.

\subsection{Contextual CwFs}\label{ssec:contextual_models}

An important class of CwFs are the \emph{contextual} CwFs, whose objects and morphisms are really given by lists of types and terms.
Indeed, from some point of view, in the language of type theory, we never explicitly talk about the objects and morphisms of a model, but only about types and terms that live in the same contextual slice of a given model.
From this perspective, only the contextual models matter.
However, a direct definition of contextual models is complicated (their generalized algebraic presentation is infinite, as it quantifies over all possible context lengths), and many intermediate constructions go through non-contextual models.
It is thus more convenient to define contextuality as a property of general models.
Fortunately, they can be described by the means of an orthogonal factorization system on $\CCwf$\footnote{The author learnt of this definition of contextuality from Christian Sattler.}.

\begin{defi}
  A CwF morphism $F : \CC \to \CD$ is said to be a \defemph{contextual isomorphism} if its actions on types and terms are bijective.
  \defiEnd{}
\end{defi}

Contextual isomorphisms are the right class of maps of the orthogonal factorization system generated by $\{I^\ty,I^\tm\}$, or equivalently of the weak factorization system generated by $\{I^\ty,E^\ty,I^\tm,E^\tm\}$.
The maps in the left class are called \defemph{left contextual maps}.

Any CwF morphism $F : \CC \to \CD$ admits an unique (up to isomorphism) factorization $\CC \to \cxlim(F) \to \CD$ where $\CC \to \cxlim(F)$ is left contextual and $\cxlim(F) \to \CD$ is a contextual isomorphism.
The CwF $\cxlim(F)$ is called the \defemph{contextual image} of $F$.

In particular, given any $\CC : \CCwf$, the unique morphism $\Init \to \CC$ admits such a factorization.
Its contextual image is called the \defemph{contextual core} of $\CC$, and is denoted by $\cxl(\CC)$.
The map $\cxl(\CC) \to \CC$ is a contextual isomorphism by definition.
When the map $\cxl(\CC) \to \CC$ is also an isomorphism of CwFs, we say that $\CC$ is \defemph{contextual}.

This definition of contextuality is equivalent to the usual definition, as found for instance in~\cite{CwFUSD}.
\begin{prop}
  A CwF $\CC$ is contextual if and only there exists a length function $l : \Ob_\CC \to \Nat$ such than for any $\Gamma \in \CC$, if $l(\Gamma) = 0$ then $\Gamma = 1_\CC$ and if $l(\Gamma) = n+1$, then there are unique $\Gamma' \in \CC$ and $(\Gamma' \vdash A\ \type) \in \CC$ such that $\Gamma = \Gamma'.A$.
  \qed
\end{prop}

We write $\CCwf^\cxl$ for the full subcategory of $\CCwf$ spanned by the contextual CwFs.
The contextual core functor $\cxl : \CCwf \to \CCwf^\cxl$ is a right adjoint of $\CCwf^\cxl$.
This exhibits $\CCwf^\cxl$ as a coreflective subcategory of $\CCwf$.

This definition of contextuality extends to $\CCwf_{\reppre}$ by considering the orthogonal factorization system generated by $\{I^\ty,E^\ty,I^{\repty},E^{\repty},I^\tm,E^\tm\}$ instead, \ie contextual isomorphisms in $\CCwf_{\reppre}$ also need to have bijective actions on representable types.

\subsection{Join of families and telescopes}

We do not always assume the presence of $\Sigma$-types.
To circumvent their absence in some constructions, we need to work with families of telescopes, whose types and terms are finite sequences of types and terms of the base family.
It is convenient to present them as the coproduct of length $n$ telescopes for all $n : \Nat$, and to generalize the notion of length $n$ telescope to a more heterogeneous notion, using the notion of join of families\footnote{The author learnt of this presentation from Christian Sattler.}.

\begin{defi}
  Let $\MC$ and $\MD$ be two families.
  Their \defemph{join} $\MC \ast \MD$ is the family defined by:
  \begin{alignat*}{1}
    &(\MC \ast \MD).\Ty \triangleq
      (A : \MC) \times (B : A \to \MD) \\
    & (\MC \ast \MD).\Tm(A, B) \triangleq
      (a : A) \times (b : B(a))
      \tag*{\defiEnd}
  \end{alignat*}
\end{defi}

In other words, the join of $\MC$ and $\MD$ is the family of length $2$ telescopes, whose first and second components come respectively from $\MC$ and $\MD$.

The join operation is a monoidal product on $\CFam$.

\begin{defi}
  If $\MC$ is a family and $n : \Nat$, the family $\MC^{\ast n}$ of \defemph{length $n$ telescopes} is the $n$-fold iterated join of $\MC$.
  \defiEnd
\end{defi}

\begin{defi}
  If $\MC$ is a family, the family of \defemph{telescopes} of $\MC$ is the coproduct $\MC^{\star} \triangleq \coprod_{n : \Nat}\MC^{\ast n}$.
  \defiEnd
\end{defi}

If $\CC$ is a contextual CwF, we may identify its objects with the closed telescopes of types (i.e. the global elements of the presheaf $\MC^\star.\Ty$) and its morphisms from $\Gamma$ to $\Delta$ with the natural transformations from $\MC^\star.\Tm(\Gamma)$ to $\MC^\star.\Tm(\Delta)$.
When $\CC$ is an arbitrary CwF, this is an explicit description of the objects and morphisms of the contextual core of $\CC$.


\section{Weak identity types}\label{sec:weak_id_types}

In this section we define the weak and strict variants of identity types.
Generally, the computation rules of a weak type structure are expressed by internal identifications, whereas the computation rules of strict type structures are expressed by strict equalities.
These notions should not be confused with weakly stable and strictly stable type-structures.
Since type theories with weak computation rules are the main focus of this paper, we reserve the short naming for this notion.
In ambiguous cases, we could talk of weakly computational and strictly computational type-structures.
In the literature, strict identity types sometimes refer to $(-1)$-truncated identity types, which satisfy the Uniqueness of Identity Proofs principle.

\subsection{Definition}

In presence of (strict) $\Sigma$-type and $\Pi$-types, there are several equivalent ways to define the eliminator for identity types.
The Martin-Löf eliminator is given by the following rule.
\begin{mathpar}
  \inferrule[Martin-Löf eliminator]{A\ \type \\ x : A, y : A, p : \Id_A(x,y) \vdash D(x,y,p)\ \type \\ x : A \vdash d(x) : D(x,x,\refl) \\ x : A \\ y : A \\ p : \Id_A(x,y)}{\J(D,d,y,p) : D(x,y,p)}
\end{mathpar}
In absence of $\Pi$-types, the Martin-Löf eliminator does not seem to be strong enough to even define transport.

%

In \cite{IdTypeFS}, a variant of the Martin-Löf eliminator, now called the Frobenius variant of the Martin-Löf eliminator, is introduced.
The idea is to circumvent the absence of $\Pi$-types by allowing the target type family $D$ of the elimination to depend on any telescope $\Delta$ of parameters.
\begin{mathpar}
  \inferrule[Frobenius Martin-Löf eliminator]{A\ \type \\ x : A, y : A, p : \Id_A(x,y) \vdash \Delta(x,y,p)\ \type^\star \\\\
    x : A, y : A, p : \Id_A(x,y), \delta : \Delta(x,y,p) \vdash D(x,y,p,\delta)\ \type \\
    x : A, \delta : \Delta(x,x,\refl) \vdash d(x,\delta) : D(x,x,\refl,\delta) \\\\
    x : A \\ y : A \\ p : \Id_A(x,y) \\ \delta : \Delta(x,y,p)}{\J(D,d,x,y,p,\delta) : D(x,y,p,\delta)}
\end{mathpar}
Another alternative is the Paulin-Mohring eliminator, also known as based path induction, or one-sided eliminator.
\begin{mathpar}
  \inferrule[Paulin-Mohring eliminator]{A\ \type \\ x : A \\ y : A, p : \Id_A(x,y) \vdash D(y,p)\ \type \\\\ d(x) : D(x,\refl) \\ y : A \\ p : \Id_A(x,y)}{\J(D,d,y,p) : D(y,p)}
\end{mathpar}
North~\cite{IdWFSCauchy} and Isaev~\cite{IsaevIxTT} have independently given proofs of the fact that the Paulin-Mohring eliminator is equivalent to the Frobenius variant of the Martin-Löf eliminator in the presence of $\Sigma$-types.

We use the weak variant of the Paulin-Mohring eliminator, with the computation rule weakened to a weak equality.
We will show that the other eliminators can also be derived, even in the absence of $\Sigma$-types.
In fact, we will prove that weak identity types can be lifted from a family $\MC$ to its telescope family $\MC^{\star}$, i.e. the family whose types are list of types of $\MC$; the derivation of the Frobenius eliminator can be seen as a consequence of this fact.
The main step of this derivation is originally due to András Kovács.
The proof has been simplified using the notion of join of families by Christian Sattler.

\begin{defi}\label{def:id_intro}
  Let $\MC$ be a family.
  An \defemph{identity type introduction structure} over $\MC$ consists of the following operations:
  \begin{alignat*}{1}
    & \Id : (A : \MC) (x : A) (y : A) \to \MC, \\
    & \refl : (A : \MC) (x : A) \to \Id_A(x,x).
      \tag*{\defiEnd}
  \end{alignat*}
\end{defi}

Given an identity type introduction structure and $A : \MC$, we write $\mathsf{IdFrom}(x)$ for the set of identifications with $x : A$ as the fixed left endpoint.
\begin{alignat*}{1}
  & \mathsf{IdFrom}(x) \triangleq (y : A) \times \Id_A(x,y).
\end{alignat*}

\begin{defi}\label{def:weak_id_elim}
  Let $\MC$ and $\MD$ be two families with identity type introduction structures.
  A \defemph{weak identity type elimination structure} from $\MC$ to $\MD$ consists of the following operations:
  \begin{alignat*}{3}
    & \J && :{ } && (A : \MC) (x : A) (D : \mathsf{IdFrom}(x) \to \MD) \\
    &&&&& (d : D(x,\refl)) (p : \mathsf{IdFrom}(x)) \\
    &&&&& \to D(p) \\
    & \J_{\beta} && :{ } && (A : \MC) (x : A) (D : \mathsf{IdFrom}(x) \to \MD) \\
    &&&&& (d : D(x,\refl)) \\
    &&&&& \to \Id(\J(D,d,(x,\refl)),d)
          \tag*{\defiEnd}
  \end{alignat*}
\end{defi}

A \defemph{weak identity type structure} over a family $\MC$ consists of an identity type introduction structure over $\MC$ along with a weak identity type elimination structure from $\MC$ to $\MC$.

\begin{defi}\label{def:strong_id}
  A weak identity type structure is said to be \defemph{strict} if the following equations hold, for all relevant arguments:
  \begin{alignat*}{1}
    & \J(D,d,(x,\refl)) = d, \\
    & \J_{\beta}(D,d) = \refl.
      \tag*{\defiEnd{}}
  \end{alignat*}
\end{defi}

Given an internal family $\MC$ equipped with weak identity types, we can derive the transport operation:
\[ \mathsf{transport} : (D : A \to \MC) (p : \Id(x,y)) \to D(x) \to D(y). \]
We may write $p^{\star}(d)$ instead of $\mathsf{transport}(D,p,d)$, leaving the family $D$ implicit.

The standard notions of homotopy type theory, such as contractible types, propositional types, equivalences, etc, can be defined.
However, since we may not have $\Sigma$-types or $\Pi$-types, they are not always encoded by types of the theory.

\begin{defi}
  An \defemph{equality type structure} over a family $\MC$ is a weak identity type structure satisfying the following two equations
  \begin{alignat*}{1}
    & \Id(A,x,y) \to (x = y), \\
    & (p : \Id(A,x,y)) \to (p = \refl).
      \tag*{\defiEnd{}}
  \end{alignat*}
\end{defi}
We will use $\Eq$ instead of $\Id$ for equality type structures.

\subsection{Lifting type structures to telescopes}\label{ssec:lift_tele}

In this subsection, we show that $\Id$-type structures on a family $\MC$ can be lifted to the telescope family $\MC^{\star}$.
Similar results have been proven and used before in the literature~\cite{2DTT, HoInvCwA}.
We generalize them to weak type structures and families without locally representable elements.
The constructions have been formalized in Agda.%
\footnote{The Agda files are available at~\url{https://rafaelbocquet.gitlab.io/Agda/20230420_WeakIdTypes/}.}
We refer the reader to the formalization for the detailed constructions.

\begin{con}\label{con:lift_wid_intro}
  Let $\MC$ and $\MD$ be families with identity type introduction structures, along with a weak identity type elimination structure from $\MC$ to $\MD$.

  Then the family $\MC \ast \MD$ is equipped with the following identity type introduction structure:
  \begin{alignat*}{1}
    & \Id_{(A,B)}((x_{c}, x_{d}), (y_{c}, y_{d})) \triangleq (\Id_A(x_c,y_c), \lambda p_{c} \mapsto \Id_{B(y_c)}(p_{c}^{\star}(y_{c}), y_{d}) \\
    & \refl_{(A,B)}(x_{c}, x_{d}) \triangleq (\refl, \mathsf{tr}_\beta),
  \end{alignat*}
  where $\mathsf{tr}_\beta$ is some term of type $\Id(\refl^{\star}(x_{c}), x_{c})$, definable using $\J_{\beta}$.
  \defiEnd
\end{con}

\begin{con}\label{con:lift_wid_join}
  Let $\MC$, $\MD$ and $\ME$ be families with identity type introduction structures, along with identity type elimination structures from $\MC$ to $\MD$ and $\ME$, from $\MD$ to $\MD$ and $\ME$ and from $\ME$ to $\ME$.

  Then there exist weak identity type elimination structures from $(\MC \ast \MD)$ to $\ME$ and from $\MC$ to $(\MD \ast \ME)$.
  \defiEnd
\end{con}

\begin{con}
  Let $\MC$ be a family with a weak identity type structure.

  Then for every $n : \Nat$, the family $\MC^{\ast n}$ of length $n$ telescopes has a canonical identity type introduction structure, and for every $n, m : \Nat$ there is a weak identity type elimination structure from $\MC^{\ast n}$ to $\MC^{\ast m}$.
\end{con}
\begin{proof}
  By iterating \cref{con:lift_wid_intro} and \cref{con:lift_wid_join}.
\end{proof}

In all of these constructions, if all of the given identity type are strict, then the constructed type structures are also strict.

\subsection{Parametrized elimination structures}\label{ssec:param_id_elim}

\begin{defi}
  Let $\MC$, $\MD$ and $\ME$ be families, together with identity type introduction structures over $\MC$ and $\ME$. \\
  A parametrized identity type elimination structure from $\MC$ to $\ME$ with parameters in $\MD$ consists of operations
  \begin{alignat*}{3}
    & \J && :{ } && (A : \MC) (x : A) (Q : \mathsf{IdFrom}(x) \to \MD) \\
    &&&&& (D : (p : \mathsf{IdFrom}(x)) (q : Q(p)) \to \ME) \\
    &&&&& (d : (q : Q(x,\refl)) \to D((x,\refl),q)) \\
    &&&&& (p : \mathsf{IdFrom}(x)) (q : Q(y)) \\
    &&&&& \to D(p,q) \\
    & \J_{\beta} && :{ } && (A : \MC) (x : A) (Q : \mathsf{IdFrom}(x) \to \MD) \\
    &&&&& (D : (p : \mathsf{IdFrom}(x)) (q : Q(y)) \to \ME) \\
    &&&&& (d : (q : Q(x,\refl)) \to D((x,\refl),q)) \\
    &&&&& \to \Id(\J(Q,D,d,(x,\refl)), d(q))
          \tag*{\defiEnd}
  \end{alignat*}
\end{defi}

\begin{con}
  Assume that $\MC$, $\MD$ and $\ME$ are families equipped with identity type introduction structures and identity type elimination structures from $\MC$ to $\MC$, $\MD$ and $\ME$, from $\MD$ to $\MD$ and $\ME$ and from $\ME$ to $\ME$.

  Then we can construct a parametrized identity type elimination structure from $\MC$ to $\ME$ with parameters in $\MD$.
  \defiEnd
\end{con}

The Frobenius variant of the Paulin-Mohring identity type eliminator is exactly a parametrized identity type elimination structure with parameters in the family of telescopes.

\subsection{CwFs with identity types}

A CwF $\CC$ is equipped with identity types if its family $\MC$ is equipped with a global weak identity type structure in $\CPsh(\CC)$.

We will write $\CCwf_\Id$ for the category of CwFs equipped with weak identity types.

In presence of both $\Pi$-types and $\Id$-types, we always assume that function extensionality is satisfied.
We write $\CCwf_{\repinfty}$ for the category of $(\reppre)$-CwFs equipped with weak identity types satisfying function extensionality.
We write $\CCwf_{\repO}$ for the category of $(\reppre)$-CwF equipped with extensional equality types.


\section{Second-order generalized algebraic theories}\label{sec:sogats}

In this section we recall the theory of second-order generalized algebraic theories.
Generalized algebraic theories (GATs) are theories specified by a collection of dependent sorts, operations and equations between them.
Second-order generalized algebraic theories (SOGATs) generalize GATs by allowing second-order operations, meaning that the arguments of an operation can bind some variables.
Most type theories can be encoded as SOGATs.
Any (first-order) GAT can also be seen as a SOGAT.
Among other examples, we will consider the SOGAT $\Th_\Id$ of weak identity types, the GAT $\Th_\CMonCat$ of monoidal categories, and their extensions to strict identity types and strict monoidal categories.

It is possible to develop the functorial semantics SOGATs in $\CCwf_{\reppre}$.
Given a SOGAT $\Th$, there is a notion of model of $\Th$ in any $(\reppre)$-CwF $\CC$, corresponding to an interpretation of the sorts and operations of $\Th$ as types and terms of $\CC$.
The SOGAT $\Th$ can be identified with its classifying $(\reppre)$-CwF, also written $\Th$.
A model of $\Th$ in $\CC$ is then identified with a $(\reppre)$-CwF morphism $\Th \to \CC$.

By working with SOGATs we know that our results apply to a large class of type theories.
More importantly, many semantic conditions can be expressed more naturally at the level of the classifying $(\reppre)$-CwF $\Th$.
Instead of considering arbitrary models of $\Th$, it is often sufficient to consider $\Th$ instead.
The reason combines two ideas.
First, any object of $\Th$ can be seen as a signature presenting a finitely generated model; this should induce a fully faithful functor $\Init_\Th[-] : \Th^\op \to \CMod_\Th$ whose essential image consists of the finitely generated models of $\Th$, see~\cref{clm:ff_sign}.
Second, any model can be approximated by finitely generated models; more precisely any model admits a cofibrant replacement that is freely generated, and any freely generated model is a filtered colimit of finitely generated models.

For this reason, we focus on the classifying $(\reppre)$-CwF of a SOGAT.
Our main results are stated at this level and do not depend on the semantics of SOGATs.
We still develop some of the semantics in order to provide some additional intuition for our definitions.
However, semantic considerations that would require rather technical proofs are omitted or left to future work.

We also review the notion of external univalence for a SOGATs equipped with homotopy relations.
This notion was introduced by the author in~\cite{ExternalUnivalence}.
Given a choice of a homotopy relation for every sort of a SOGAT, external univalence expresses the well-behavedness of the homotopy relations; in particular the fact that they are preserved by all operations.
External univalence is also stated at the level of the classifying $(\reppre)$-CwF $\Th$, as the existence of some identity type structure (unfortunately both weakly computational and weakly stable) on $\Th$.
These identity types have to satisfy a univalence condition, also called saturation condition.
An analogous partial saturation condition will be introduced in~\cref{sec:partial_saturation}.

More semantically, external univalence is closely related to left semi-model structures on the category of models, as considered by Kapulkin and Lumsdaine~\cite{HoThTT} and Isaev~\cite{IsaevMS}.
External univalence is also a way to morally view a SOGAT as an $\infty$-type theory (in the sense of Nguyen and Uemura~\cite{InftyTypeTheories}), without using quasicategories.

\subsection{Definitions and examples}

\begin{defi}\label{def:sogat}
  A \defemph{second-order generalized algebraic theory} is an $\{I^{\ty},I^{\repty},I^{\tm},E^{\tm}\}$-cellular $(\reppre)$-CwF $\Th$, where we recall that the maps $\{I^{\ty},I^{\repty},I^{\tm},E^{\tm}\}$ are the generic extensions of $(\reppre)$-CwFs by a type, representable type, term or term equality:
  \begin{alignat*}{3}
    & I^{\ty} && :{ }
    && \Free_{\reppre}(\bm{\Gamma} \vdash) \to \Free_{\reppre}(\bm{\Gamma} \vdash \bm{A}\ \type),
    \\
    & I^{\repty} && :{ }
    && \Free_{\reppre}(\bm{\Gamma} \vdash) \to \Free_{\reppre}(\bm{\Gamma} \vdash \bm{A}\ \reptype),
    \\
    & I^{\tm} && :{ }
    && \Free_{\reppre}(\bm{\Gamma} \vdash \bm{A}\ \type) \to \Free_{\reppre}(\bm{\Gamma} \vdash \bm{a} : \bm{A}),
    \\
    & E^{\tm} && :{ }
    && \Free_{\reppre}(\bm{\Gamma} \vdash \bm{x},\bm{y} : \bm{A}) \to \Free_{\reppre}(\bm{\Gamma} \vdash \bm{x} = \bm{y}).
       \tag*{\defiEnd{}}
  \end{alignat*}
\end{defi}
A morphism between SOGATs is a morphism of $(\reppre)$-CwFs.

The cellular presentation of a SOGAT consists of sorts, representable sorts, operations and equations between operations, corresponding to the four maps $I^{\ty},I^{\repty},I^{\tm},E^{\tm}$. 

A SOGAT without equations is a SOGAT that is actually $\{I^\ty,I^{\repty},I^{\tm}\}$-cellular in $\CCwf_{\reppre}$.
SOGATs without equations are cofibrant objects in $\CCwf_{\reppre}$, and are therefore better behaved than general SOGATs in some situations.
When considering equational extensions of SOGATs, the source SOGAT will usually be a SOGAT without equations.

We will write $\MT$, $\MT_1$, $\MT_2^3$, \etc, for the internal $(\reppre)$-families corresponding to SOGATs $\Th$, $\Th_1$, $\Th_2^3$, \etc in the presheaf categories $\CPsh(\Th)$, $\CPsh(\Th_1)$, \etc.

We use a {\color{RedOrange}\bf bold red} font for the sorts and operations of a SOGAT, mainly to distinguish them from the type-theoretic structure of the $(\reppre)$-CwF.

\begin{exa}\label{exa:gat_cat}
  The GAT of categories $\Th_{\CCat}$ is given by the following signature:
  \begingroup{}\allowdisplaybreaks{}
  \begin{alignat*}{1}
    & \iOb : \MT_{\CCat} \\
    & \iHom : \iOb \to \iOb \to \MT_{\CCat} \\
    & \iEqHom : \forall x\ y \to \iHom(x,y) \to \iHom(x,y) \to \MT_{\CCat} \\
    & \iid :{ } \forall x \to \iHom(x,x) \\
    & \icomp :{ } \forall x\ y\ z \to \iHom(x,y) \to \iHom(y,z) \to \iHom(x,z) \\
    & f \icirc g \triangleq{ } \icomp(g,f) \\
    & \irefl :{ } \forall x\ y\ \to (f : \iHom(x,y)) \to \iEqHom(f,f) \\
    & \iid \icirc f = f \\
    & f \icirc \iid = f \\
    & (f \icirc g) \icirc h = f \icirc (g \icirc h) \\
    & \forall x\ y\ f\ g \to (p,q : \iEqHom(f,g)) \to p = q \\
    & \iEqHom(f,g) \to f = g
      \tag*{\defiEnd{}}
  \end{alignat*}\endgroup{}
\end{exa}

Some of the equations of $\Th_{\CCat}$ can be replaced by operations targeting the sort $\iEqHom(f,g)$.
However, some equations cannot be avoided, namely the propositionality of $\iEqHom$ and the reflection rule $\iEqHom(f,g) \to f = g$.

An alternative to the theory of categories is the theory of $E$-categories~\cite{AczelGalois,PalmgrenCategories}, which does not include equations.
An $E$-category is essentially a category enriched over setoids, \ie with a setoid of morphisms between any two objects, except that the categorical laws are expressed using the relations of the setoids.

\begin{exa}
  The GAT $\Th_{\ECat} $ of $E$-categories is presented by the following signature:
  \begingroup{}\allowdisplaybreaks{}
  \begin{alignat*}{1}
    & \iOb : \MT_{\ECat} \\
    & \iHom : \iOb \to \iOb \to \MT_{\ECat} \\
    & \iEqHom : \forall x\ y \to \iHom(x,y) \to \iHom(x,y) \to \MT_{\ECat} \\
    & \iid : \iHom(x,x) \\
    & \icomp : \iHom(x,y) \to \iHom(y,z) \to \iHom(x,z) \\
    & f \icirc g \triangleq \icomp(g,f) \\
    & \irefl : (f : \iHom(x,y)) \to \iEqHom(f,f) \\
    & \inner{sym} : \iEqHom(f,g) \to \iEqHom(g,f) \\
    & \inner{trans} : \iEqHom(f,g) \to \iEqHom(g,h) \to \iEqHom(f,h) \\
    & \inner{assoc} : \iEqHom(f \icirc (g \icirc h), (f \icirc g) \icirc h) \\
    & \inner{idl} : \iEqHom(\iid \icirc f, f) \\
    & \inner{idr} : \iEqHom(f \icirc \iid, f) \\
    & \inner{cong\text{-}\icirc} : \iEqHom(f, h) \to \iEqHom(g, i) \to \iEqHom(f \icirc g, h \icirc i) 
  \end{alignat*}\endgroup{}

  The GAT $\Th_{\CCat}$ is isomorphic to the extension of $\Th_{\ECat}$ by the equations
  \begin{alignat*}{1}
    & \iEqHom(f,g) \to f = g, \\
    & (p : \iEqHom(f,g)) \to p = \irefl.
      \tag*{\defiEnd{}}
  \end{alignat*}
\end{exa}

\begin{exa}\label{exa:gat_moncat}
  The GAT $\Th_{\CMonCat}$ of (weak) monoidal categories is the extension of $\Th_{\CCat}$ by the following operations:
  \begingroup{}\allowdisplaybreaks{}
  \begin{alignat*}{1}
    & \iI :{ } \iOb, \\
    & \_ \iotimes \_ :{ } \iOb \to \iOb \to \iOb, \\
    & \_ \iotimes \_ :{ } \iHom(x,x') \to \iHom(y,y') \to \iHom(x \iotimes y, x' \iotimes y'), \\
    & \inner{\alpha} :{ } x \iotimes (y \iotimes z) \icong (x \iotimes y) \iotimes z, \\
    & \inner{\lambda} :{ } x \iotimes I \icong x, \\
    & \inner{\rho} :{ } I \iotimes x \icong x,
  \end{alignat*}\endgroup{}
  and equations (or operations returning in $\iEqHom$) expressing the functoriality of $\iotimes$, the naturality of $\alpha$, $\lambda$ and $\rho$, and the pentagon and triangle coherence laws.
  \defiEnd{}
\end{exa}

\begin{exa}\label{exa:gat_strmoncat}
  The GAT $\Th_{\CStrMonCat}$ of strict monoidal categories is the extension of $\Th_{\CMonCat}$ by the following equations:
  \begin{mathpar}
    x \iotimes (y \iotimes z) = (x \iotimes y) \iotimes z,

    \inner{\alpha}_{x,y,z} = \id,
    \\
    x \otimes I = x,

    \inner{\lambda}_x = \id,
    
    I \otimes x = x,

    \inner{\rho}_x = \id.
  \end{mathpar}
\end{exa}

\begin{exa}
  The SOGAT $\Th_\Id$ of weak identity types is specified by the following signature:
  \begin{alignat*}{1}
    & \iTy : \MT_{\Id}, \\
    & \iTm : \iTy \to {(\MT_{\Id})}_{\rep}, \\
    & \iId : (A : \iTy) (x : \iTm(A)) (y : \iTm(A)) \to \iTy, \\
    & \irefl : (A : \iTy) (x : \iTm(A)) \to \iTm(\iId(A,x,x)), \\
    & \iJ : (A : \iTy) (x : \iTm(A)) (P : (y : \iTm(A)) (p : \iTm(\iId(A,x,y))) \to \iTy) \\
    & \quad (d : \iTm(P(x,\irefl))) (y : \iTm(A)) (p : \iTm(\iId(A,x,y))) \to \iTm(P(y,p)), \\
    & \iJb : (A : \iTy) (x : \iTm(A)) (P : (y : \iTm(A)) (p : \iTm(\iId(A,x,y))) \to \iTy) \\
    & \quad (d : \iTm(P(x,\irefl))) \to \iTm(\iId(P(x,\refl),\iJ(A,x,P,d,x,\irefl),d)).
  \end{alignat*}

  Note that $\Th_\Id$ is a SOGAT without equations.
  \defiEnd{}
\end{exa}

\begin{exa}
  The SOGAT $\Th_{\Id_s}$ of strict identity types is the extension of $\Th_\Id$ by equations
  \begin{mathpar}
    \iJ(A,x,P,d,x,\irefl) = \irefl, 
    
    \iJb = \irefl.
  \end{mathpar}
  \defiEnd{}
\end{exa}

\subsection{Structure of the types of a SOGAT}\label{ssec:sogat_structure_types}

Because SOGATs are $(\reppre)$-CwFs that are presented without equations between types or representable types, their types admit simple normal forms: they are the free closure of the generating types by the operations of a $(\reppre)$-CwF.

There is a set $\GenTy_\Th$ of generating types, which can be read from the cellular presentation of $\Th$.
For example, the generating types of $\Th_\CCat$ are $(\vdash \iOb\ \type)$, $(x,y:\iOb \vdash \iHom(x,y)\ \type)$ and $(x,y:\iOb,f,g:\iHom(x,y) \vdash \iEqHom(f,g)\ \type)$.

Similarly, there is a set $\GenRepTy_\Th$ of generating representable types, with an inclusion $\GenRepTy_\Th \hookrightarrow \GenTy_\Th$.

A \defemph{basic type} of $\Th$ is a type of the form $(\Gamma \vdash \iS(\sigma)\ \type)$, where $\sigma : \Gamma \to \partial \iS$ for a generating type $(\partial \iS \vdash \iS) \in \GenTy_\Th$.
The \defemph{basic representable types} are defined similarly.
A basic non-representable type is a basic type whose generating type is in $(\GenTy_\Th \backslash \GenRepTy_\Th)$.

\begin{prop}
  Any type or representable type in $\Th$ can be uniquely obtained from the basic non-representable types and basic representable types, and from the operations $\Unit$-, $\Sigma$-, $\Pi_\rep$ and the inclusion from representable types into types.
  \qed
\end{prop}

\subsection{Functorial semantics}

\begin{defi}
  An \defemph{internal model} of $\Th$ in a $(\reppre)$-CwF $\CC$ is a $(\reppre)$-CwF morphism
  \[ \MC : \Th \to \CC.
    \tag*{\defiEnd{}}
  \]
\end{defi}
We often write $X_{\CC}$, $A_{\CC}$, $a_{\CC}$, \etc{}, or even just $X$, $A$, $a$, \etc{} instead of $\MC(X)$, $\MC(A)$, $\MC(a)$, \etc{} for the application of the $(\reppre)$-CwF morphism $\MC$ on objects, morphisms, types and terms.

By the universal property of $\Th$, an internal model in $\CC$ is uniquely determined by the image of the generators of $\Th$, that is by an interpretation of the signature $\Th$ in $\CC$.

\begin{defi}
  A \defemph{model} of $\Th$ consists of a category $\CC$ with a terminal object, along with an internal model of $\Th$ in the $(\reppre)$-CwF $\widehat{\CC}$, that is a $(\reppre)$-CwF morphism $\MC : \Th \to \widehat{\CC}$.
  \defiEnd{}
\end{defi}

\begin{exa}
  A model of the SOGAT $\Th_\Id$ of weak identity types is exactly a CwF equipped with weak identity types.

  However, the models of the SOGAT $\Th_{\CCat}$ are not categories, but pairs $(\CC,\CD)$ where $\CC$ is a base category and $\CD$ is an internal category in $\CPsh(\CC)$.
  In~\cref{ssec:sogat_contextual} we will define contextual models.
  In the case of $\Th_\CCat$ we recover categories: the contextual models of $\Th_\CCat$ are the models $(\CC,\CD)$ such that $\CC$ is the terminal category $1_\CCat$.
  $\CD$ is then an internal category in $\CPsh(1_\CCat)$, \ie an external category.
  \defiEnd{}
\end{exa}

Given a model $\CC$ of $\Th$ and a closed sort $X \in \Th$, we will write $(\Gamma \vdash x : X) \in \CC$ to indicate that $x : X_\CC(\Gamma)$.
Given a dependent sort $(X \vdash Y\ \type) \in \Th$, we will write $(\Gamma \vdash y : Y(x)) \in \CC$ to indicate that $y : Y_\CC(\Gamma,x)$.
We may also write $(\gamma:\Gamma \vdash x(\gamma) : X) \in \CC$ and $(\gamma:\Gamma \vdash y(\gamma) : Y(x(\gamma))) \in \CC$.

\begin{defi}
  A \defemph{weak morphism} $F$ of models of $\Th$ consists of a functor $F : \CC \to \CD$ such that:
  \begin{itemize}
    \item The functor $F$ weakly preserves terminal objects.
    \item For every object $X \in \Th$, we have a transformation
          \[ F_X : (\Gamma \in \CC) \to X_\CC(\Gamma) \to X_\CD(F(\Gamma)) \]
          contravariantly natural in $\Gamma$.
    \item For every morphism $\alpha : X \to Y$, the following square commutes
          \[ \begin{tikzcd}
              X_{\CC}(\Gamma)
              \ar[d, "\alpha_{\CC}(\Gamma)"]
              \ar[r, "F_X"] &
              X_\CD(F(\Gamma))
              \ar[d, "\alpha_\CD(F(\Gamma))"] \\
              Y_\CC(\Gamma)
              \ar[r, "F_Y"] &
              Y_\CD(F(\Gamma))
          \end{tikzcd} \]
    \item Remark that we obtain, for every object $X \in \Th$ and type $(X \vdash A\ \type) \in \Th$, a transformation
          \begin{alignat*}{1}
            & F_A :  (\Gamma \in \CC) \to (x : X_{\CC}(\Gamma)) \to A_{\CC}(x) \to A_{\CD}(F_X(x))
          \end{alignat*}
          contravariantly natural in $\Gamma$, uniquely specified by the equation $F_{X.A}(x,a) = (F_{X}(x),F_{A}(x,a))$.
    \item Context extensions are weakly preserved: for every object $X \in \Th$, representable type $(X \vdash A\ \type_\rep) \in \Th$, object $\Gamma \in \CC$ and element $x : X_\CC(\Gamma)$, the comparison map
          \[ \angles{F(\bm{p}),F_A(\bm{q})} : F(\Gamma.A_{\CC}[x]) \to F(\Gamma).A_{\CD}[F_{X}(x)] \]
          is an isomorphism.
  \end{itemize}

  A morphism is \defemph{strict} if the terminal object and context extensions are strictly preserved.
  \defiEnd{}
\end{defi}

We have a $(2,1)$-category of models, weak morphisms and $2$-cells, and a $1$-category $\CMod_{\Th}$ of models and strict morphisms.
We will only work with the $1$-category $\CMod_{\Th}$.
The category $\CMod_{\Th}$ is locally finitely presentable; in particular we have an initial model $\Init_{\Th}$ and more general freely generated models.

\subsection{Contextuality and trivial fibrations}\label{ssec:sogat_contextual}

Similarly to our definitions of trivial fibrations and contextuality for CwFs, we can define these notions for models of arbitrary SOGATs using factorization systems on the category of models.

We define a set $I = \{ I^\iS \mid \iS \in \GenTy_\Th \}$ of maps in $\CMod_{\Th}$, where
\[ I^\iS : \Free_\Th(\bm{\Gamma} \vdash \bm{\sigma} : \partial \iS) \to \Free_\Th(\bm{\Gamma} \vdash \bm{x} : \iS(\bm{\sigma})) \]
is the generic extension of a model by an element of the sort $\iS$.

We write $E = \{ E^\iS \mid \iS \in \GenTy_\Th \}$ for the set of codiagonals of maps in $I$.
\[ E^\iS : \Free_\Th(\bm{\Gamma} \vdash \bm{x},\bm{y} : \iS(\bm{\sigma})) \to \Free_\Th(\bm{\Gamma} \vdash \bm{x} = \bm{y}) \]
is the generic extension of a model by a new equation between elements of the sort $\iS$.

\begin{defi}
  A morphism $F : \CC \to \CD$ in $\CMod_\Th$ is a \defemph{split trivial fibration} if for every generating sort $\iS \in \GenTy_\Th$ the actions of $F$ on the elements of $\iS$ are split surjections: for every object $\Gamma \in \CC$, boundary $(\gamma:\Gamma \vdash \sigma(\gamma) : \partial \iS) \in \CC$ and element $(\gamma:F(\Gamma) \vdash x(\gamma) : \iS(F(\sigma)(\gamma))) \in \CD$, there is a lifted element $(\gamma:\Gamma \vdash x_0(\gamma) : \iS(\sigma(\gamma))) \in \CC$ such that $F(x_0) = x$.
  \defiEnd{}
\end{defi}

The trivial fibrations are the right class of maps in the weak factorization system generated by $I$.
The maps in the left class are called \defemph{cofibrations}.

There is also a notion of non-split trivial fibration.
\begin{defi}
  A morphism $F : \CC \to \CD$ in $\CMod_\Th$ is a \defemph{non-split trivial fibration} if for every generating sort $\iS \in \GenTy_\Th$ the actions of $F$ on the elements of $\iS$ are non-split surjections.
  \defiEnd{}
\end{defi}

Finally, contextuality is specified using the notion of contextual isomorphism.
\begin{defi}
  A morphism $F : \CC \to \CD$ in $\CMod_\Th$ is a \defemph{contextual isomorphism} if for every generating sort $\iS \in \GenTy_\Th$, the actions of $F$ on the elements of $\iS$ are bijections.
  \defiEnd{}
\end{defi}
The contextual isomorphisms are the right class of maps of the orthogonal factorization system generated by $I$, or equivalently in the weak factorization system generated by $I \cup E$.

Any morphism $F : \CC \to \CD$ admits a unique (up to isomorphism) factorization $\CC \to \cxlim(F) \to \CD$ where $\CC \to \cxlim(F)$ is a $(I \cup E)$-cofibration and $\cxlim(F) \to \CD$ is a contextual isomorphism.
The model $\cxlim(F)$ is called the \defemph{contextual image} of $F$.

The \defemph{contextual core} $\cxl(\CC)$ of a model $\CC$ is the contextual image of the unique morphism $\Init_\Th \to \CC$.
The map $\cxl(\CC) \to \CC$ is a contextual isomorphism by definition.
When the map $\cxl(\CC) \to \CC$ is also an isomorphism in $\CMod_\Th$, we say that $\CC$ is \defemph{contextual}.

\begin{conj}\label{clm:ff_sign}
  There is a functor
  \[ \Init_\Th[-] : \Th^\op \to \CMod_\Th^\cxl. \]
  It is fully faithful, and its essential image consists of the finite $I$-cellular models of $\Th$.
  \defiEnd{}
\end{conj}

\subsection{Homotopy relations}

The data of homotopy relations on a SOGAT equips every sort of the SOGAT with a notion of identification or sameness that is weaker than definitional equality.
It does not presuppose that these relations are well-behaved, for example they could fail to even be symmetric and transitive.
The well-behavedness of homotopy relations is encoded in the notion of external univalence.

\begin{defi}
  The data of \defemph{homotopy relations} on a SOGAT $\Th$ consists, for every generating sort $\iS$, of a reflexive type-valued relation on its elements:
  \begin{alignat*}{1}
    & \_ \sim_{\iS(\_)} \_ : (\sigma : \partial \iS) (x,y : \iS(\sigma)) \to \MT, \\
    & \hrefl_{\iS(\_)} : (\sigma : \partial \iS) (x : \iS(\sigma)) \to x \sim_{\iS(\sigma)} x.
      \tag*{\defiEnd{}}
  \end{alignat*}
\end{defi}
We use $\hrefl$ for the constant homotopy, to distinguish it from reflexivity witnesses of other relations, such as $\irefl$ in $\Th_\Th$, or $\refl$ in a $(\repinfty)$-CwF.

\begin{exa}\label{exa:homotopy_relations_cat}
  The theory $\Th_{\ECat}$ of $E$-categories is equipped with the following homotopy relations:
  \begin{alignat*}{1}
    & (x \sim_{\iOb} y) \triangleq (x \icong y), \\
    & \hrefl_\iOb(x) \triangleq \iid, \\
    & (f \sim_{\iHom(x,y)} g) \triangleq \iEqHom(f,g), \\
    & \hrefl_{\iHom(x,y)}(f) \triangleq \irefl(f), \\
    & (p \sim_{\iEqHom(f,g)} q) \triangleq \Unit \\
    & \hrefl_{\iEqHom(f,g)}(p) \triangleq \star,
  \end{alignat*}
  where $(x \icong y)$ is the sort of isomorphisms between $x$ and $y$.
  \begin{alignat*}{2}
    & (x \icong y)
    &&{ }\triangleq (f : \iHom(x,y)) \times (g : \iHom(y,x)) \\
    &&& \quad \times \iEqHom(f \icirc g, \iid) \times \iEqHom(g \icirc f, \iid).
  \end{alignat*}

  The theories $\Th_{\CCat}$, $\Th_{\CMonCat}$ and $\Th_{\CStrMonCat}$, which can be seen as extensions of $\Th_{\ECat}$, are equipped with the same homotopy relations.
  \defiEnd{}
\end{exa}

\begin{exa}\label{exa:homotopy_relations_id}
  The theory $\Th_{\Id}$ of weak identity types is equipped with the following homotopy relations:
  \begin{alignat*}{1}
    & (A \sim_{\iTy} B) \triangleq (A \simeq B), \\
    & \hrefl_\iTy(A) \triangleq \iId_A, \\
    & (x \sim_{\iTm(A)} y) \triangleq \iTm(\iId_A(x,y)), \\
    & \hrefl_{\iTm(A)}(x) \triangleq \irefl(x).
  \end{alignat*}
  where $(x \simeq y)$ is the sort of equivalences between $A$ and $B$.
  Equivalences are defined to be type-valued relations that are functional in both directions:
  \begin{alignat*}{2}
    & (A \simeq B)
    &&{ }\triangleq (R : A \to B \to \iTy) \\
    &&& \quad \times ((a : A) \to \isContr((b:B) \times R(a,b)) \\
    &&& \quad \times ((b : B) \to \isContr((a:A) \times R(a,b)).
  \end{alignat*}
  Note that $(b:B) \times R(a,b)$ and $(a:A) \times R(a,b)$ are not types (elements of $\iTy$) in the absence of $\Sigma$-types, but rather telescopes of length $2$.
  Contractibility can be defined for telescopes using the results of~\cref{sec:weak_id_types}.

  In that setting the relation specified by the identity type $\iId_A$ is the identity equivalence.
  \defiEnd{}
\end{exa}

The data of homotopy relations is sufficient to equip $\CMod_\Th$ with a notion of weak equivalences, trivial cofibrations and fibrations.
Together with the previously defined (cofibrations, trivial fibrations) weak factorization system, the category $\CMod_\Th$ is equipped with all the data of a model category, but may still fail to satisfy its axioms.

\begin{defi}
  A morphism $F : \CC \to \CD$ of models of $\Th$ is a \defemph{weak equivalence} if it satisfies the following weak lifting property for every generating sort $\iS \in \GenTy_\Th$:
  \begin{description}
  \item[Weak lifting] For every object $\Gamma \in \CC$, boundary $(\gamma:\Gamma \vdash \sigma(\gamma) : \partial \iS) \in \CC$ and term $(\gamma:F(\Gamma) \vdash x(\gamma) : \iS(F(\sigma)(\gamma))) \in \CD$, there is a lifted term
    \[ (\gamma:\Gamma \vdash x_0(\gamma) : \iS(\sigma(\gamma))) \in \CC \]
    and a homotopy
    \[ (\gamma:F(\Gamma) \vdash p(\gamma) : F(x_0)(\gamma) \sim_{\iS(F(\sigma)(\gamma))} x(\gamma)) \in \CD. 
      \tag*{\defiEnd{}}
      \]
  \end{description}
\end{defi}

As with trivial fibrations, we actually have two notions of weak equivalences: split weak equivalences, which come with the data of weak lifts for every lifting problem, and non-split weak equivalences, for which only the mere existence of lift is assumed.

\begin{exa}
  In categories seen as contextual models of $\Th_\CCat$, the weak equivalences are the functors that are essentially surjective on objects, morphisms and equalities between morphisms, \ie functors that are essentially surjective on objects and fully faithful.
  The weak equivalences in $\Th_\CMonCat$ and $\Th_\CStrMonCat$ are the same.

  Weak equivalences in $\Th_\Id$ are CwF morphisms $F : \CC \to \CD$ that are essentially surjective on types and terms, \ie that satisfy the following two weak lifting conditions:
  \begin{description}
  \item[weak type lifting] For every $\Gamma \in \CC$ and $(F(\Gamma) \vdash A\ \type) \in \CD$, there is some $(\Gamma \vdash A_0\ \type) \in \CC$ and an equivalence $\alpha : F(A_0) \simeq A$ in $\CD$.
  \item[weak term lifting] For every $(\Gamma \vdash A\ \type) \in \CC$ and $(F(\Gamma) \vdash a : F(A)) \in \CD$, there is some $(\Gamma \vdash a_0 : A) \in \CC$ and an identification $p : F(a_0) \simeq_{F(A)} a$ in $\CD$.
  \end{description}
  \defiEnd{}
\end{exa}

\begin{defi}
  A morphism $F : \CC \to \CD$ of models of $\Th$ is a \defemph{fibration} if it satisfies the following path lifting property for every generating sort $\iS \in \GenTy_\Th$:
  \begin{description}
  \item[Path lifting] For every object $\Gamma \in \CC$, boundary $(\gamma:\Gamma \vdash \sigma(\gamma) : \partial \iS) \in \CC$, terms $(\gamma:F(\Gamma) \vdash x(\gamma),y(\gamma) : \iS(F(\sigma)(\gamma))) \in \CD$ and homotopy $(\gamma : F(\Gamma) \vdash p(\gamma) : F(x)(\gamma) \sim_{\iS(F(\sigma)(\gamma))} y(\gamma)) \in \CD$, there is a lifted homotopy
    \[ (\gamma:\Gamma \vdash p_0(\gamma) : x(\gamma) \sim_{\iS(\sigma(\gamma))} y_0(\gamma)) \in \CC \]
    such that $F(y_0) = y$ and $F(p_0) = p$.
  \end{description}
\end{defi}
These fibrations could be called right fibrations, since their lifting property only involves homotopies with a fixed left endpoint.
We could also consider a notion of left fibration.

The fibrations are the right class of maps of a weak factorization system generated by $J = \{J^\iS \mid \iS \in \GenTy_\Th \}$, where
\[ J^\iS : \Free_\Th(\bm{\Gamma} \vdash \bm{x} : \iS(\bm{\sigma})) \to \Free_\Th(\bm{\Gamma} \vdash \bm{p} : \bm{x} \sim_{\iS(\bm{\sigma})} \bm{y}). \]
Maps in the left class are called \defemph{trivial cofibrations}.

\begin{exa}
  Fibrations in categories seen as contextual models of $\Th_\CCat$ are isofibrations: functors $F : \CC \to \CD$ such that for every isomorphism $g : F(x) \cong y$ in $\CD$, there is an isomorphism $g_0 : x \cong y_0$ in $\CC$ such that $F(y_0) = y$ and $F(g_0) = g$.

  Fibrations in models of $\Th_\Id$ are CwF morphisms that satisfy lifting properties for equivalences and identifications.
  \defiEnd{}
\end{exa}

Without further assumptions, these classes of maps may fail to satisfy the axioms of a model structure (see~\cref{clm:ext_univ_model_structure}).
We won't need all axioms in this paper, but we need some directions of the $2$-out-of-$3$ property for weak equivalences.

\begin{defi}
  We say that the homotopy relations on $\Th$ are \defemph{transitive} if for every generating sort $\iS \in \GenTy_\Th$, we can derive
  \[ \forall \sigma\ x\ y\ z \to (x \sim_{\iS(\sigma)} y) \to (y \sim_{\iS(\sigma)} z) \to (x \sim_{\iS(\sigma)} z) \]
  in $\Th$.
  \defiEnd{}
\end{defi}

\begin{lem}\label{lem:modt_two_out_of_three_2}
  Let $\CC \xrightarrow{F} \CD \xrightarrow{G} \CE$ be a composable pair of morphisms.
  If both $G$ and $G \circ F$ are weak equivalences, then $F$ is also a weak equivalence.
\end{lem}
\begin{proof}
  Let $\iS$ be a generating sort of $\Th$ and take $(\gamma:\Gamma \vdash \sigma(\gamma) : \partial \iS) \in \CC$ and $(\gamma:F(\Gamma) \vdash a(\gamma) : \iS(F(\sigma)(\gamma))) \in \CD$.
  Since $G \circ F$ is a weak equivalence, we can find a lift $(\gamma:\Gamma \vdash a_0(\gamma) : \iS(\sigma(\gamma))) \in \CC$ and a homotopy $(\gamma:G(F(\Gamma)) \vdash p(\gamma) : G(F(a_0))(\gamma) \sim_\iS G(a)(\gamma)) \in \CE$.
  Since $G$ is a weak equivalence, we can find a lift $(\gamma:F(\Gamma) \vdash p_0(\gamma) : F(a_0)(\gamma) \sim_\iS a(\gamma))$ of the homotopy $p$.
  Then $p_0$ witnesses the fact that $a_0$ is a weak lift of $a$.
\end{proof}

\begin{lem}\label{lem:modt_two_out_of_three_1}
  Let $\CC \xrightarrow{F} \CD \xrightarrow{G} \CE$ be a composable pair of morphisms.
  Assume that the homotopy relations of $\Th$ are transitive.
  If both $G$ and $F$ are weak equivalences, then $G \circ F$ is also a weak equivalence.
\end{lem}
\begin{proof}
  Let $\iS$ be a generating sort of $\Th$ and take $(\gamma:\Gamma \vdash \sigma(\gamma) : \partial \iS) \in \CC$ and $(\gamma:G(F(\Gamma)) \vdash a(\gamma) : \iS(G(F(\sigma))(\gamma))) \in \CE$.
  Since $G$ is a weak equivalence, we can find a lift $(\gamma:F(\Gamma) \vdash a_0(\gamma) : \iS(F(\sigma)(\gamma))) \in \CD$ and a homotopy $(\gamma:G(F(\Gamma)) \vdash p_0(\gamma) : G(a_0)(\gamma) \sim_\iS a(\gamma)) \in \CE$.
  Since $F$ is a weak equivalence, we can find a lift $(\gamma:\Gamma \vdash a_1(\gamma) : \iS(\sigma(\gamma))) \in \CC$ and a homotopy $(\gamma:F(\Gamma) \vdash p_1(\gamma) : F(a_0)(\gamma) \sim_\iS a_0(\gamma)) \in \CD$.
  By transitivity, we obtain a homotopy $(\gamma:G(F(\Gamma)) \vdash G(p_0)(\gamma) \cdot p_1(\gamma) : G(F(a_1))(\gamma) \sim_\iS a(\gamma)) \in \CE$.
  This proves that $a_1$ is a weak lift of $a$ along $G \circ F$.
\end{proof}

\subsection{External univalence}

The notion of external univalence then captures the well-behavedness of homotopy relations on a SOGAT.

Here a $(\reppre,\Id_{\mathsf{ws}})$-CwF is a $(\reppre)$-CwF equipped with weakly stable identity types satisfying function extensionality.
We have not defined weakly stable identity types, but they are only needed for this definition, which is only used to define the notion of external univalence.
\begin{defi}\label{def:univalent_internal_model}
  Let $\CC$ be a $(\reppre,\Id_{\mathsf{ws}})$-CwF equipped with an internal model of $\Th$, \ie a $(\reppre)$-CwF morphism $\Th \to \CC$.
  We say that $\CC$ is \defemph{saturated} (or that the internal model is \defemph{univalent}) if for every generating sort $\iS$, the dependent type
  \[ (y : \iS(\sigma)) \times (p : x \sim_{\iS(\sigma)} y) \]
  is contractible over the context $(\sigma : \partial \iS, x : \iS(\sigma))$.
  \defiEnd{}
\end{defi}
Note that the above definition also makes sense for $(\repinfty)$-CwFs.

\begin{defi}
  We say that $\Th$ satisfies \defemph{external univalence} if $\Th$ can be equipped with weakly stable identity types that satisfy function extensionality and saturation.
  \defiEnd{}
\end{defi}

The following conjecture was made in~\cite{ExternalUnivalence}.
\begin{conj}\label{clm:ext_univ_model_structure}
  A SOGAT $\Th$ satisfies external univalence if and only if the category $\CMod_\Th^\cxl$, equipped with the previously defined classes of cofibrations, weak equivalences and fibrations, is a left-semi model category.
  \defiEnd{}
\end{conj}

\begin{exa}
  The theories $\Th_\CCat$, $\Th_\CMonCat$ and $\Th_\CStrMonCat$ of categories, monoidal categories and strict monoidal categories all satisfy external univalence with respect to the homotopy relations defined in~\cref{exa:homotopy_relations_cat}.

  Indeed, the fact that $\Th_\CCat$ satisfies external univalence is proven in~\cite{ExternalUnivalence}.
  Using the methods presented in that paper, in order to prove that $\Th_\CMonCat$ and $\Th_\CStrMonCat$ satisfy external univalence, it suffices to additionally show that the additional operations of $\Th_\CMonCat$ and $\CStrMonCat$ preserve isomorphisms in a way that preserves identities.
  This only needs to be checked for $\iI$ and $\iotimes$; it is immediate for $\iI$ and follows from the bifunctoriality for $\iotimes$.
  \defiEnd{}
\end{exa}

\begin{exa}
  The theory $\Th_\Id$ of weak identity types satisfies external univalence for the homotopy relations of~\cref{exa:homotopy_relations_id}.
  \defiEnd{}
\end{exa}


\section{Equational extensions and conservativity}\label{sec:equational_extensions}

\subsection{Equational extensions}

We now define the notion of equational extension of a SOGAT.
Let $\Th$ be a SOGAT equipped with homotopy relations.

An equational extension of $\Th$ is an extension $\Th \to \Th_E$ of SOGATs, where $E$ is a family of homotopies in $\Th$.
In $\Th_E$, all the homotopies in $E$ become definitional equalities.
Here, a homotopy in $\Th$ is a tuple $(\iS,\Gamma,\sigma,x,y,p)$, where $\iS \in \GenTy_\Th$ is a generating sort and $(\gamma:\Gamma \vdash p(\gamma) : x(\gamma) \sim_{\iS(\sigma(\gamma))} y(\gamma)) \in \Th$ is a homotopy between $x$ and $y$ over $\Gamma$.

\begin{defi}
  Let $E$ be a family of homotopies of $\Th$; elements of $E$ are said to be \defemph{marked}.

  The \defemph{equational extension} of $\Th$ by $E$ is the SOGAT $\Th_E$ defined by extending $\Th$ by equations $(\gamma:\Gamma \vdash x = y)$ and
  $(\gamma:\Gamma \vdash p(\gamma) = \hrefl_{\iS(\sigma(\gamma))}(x(\gamma)))$
  for every marked homotopy $(\gamma:\Gamma \vdash p(\gamma) : x(\gamma) \sim_{\iS(\sigma(\gamma))} y(\gamma)) \in E$.
  \defiEnd
\end{defi}

The following are examples of equational extensions of first-order GATs.
\begin{exas} \hfill
  \begin{enumerate}
  \item 
    The extension $\Th_{\ECat} \to \Th_{\CCat}$ is the equational extension obtained by marking $p : \iEqHom(f,g)$ over the context $(x,y:\iOb,f,g:\iHom(x,y),p:\iEqHom(f,g))$.

  \item 
    The extension $\Th_{\CMonCat} \to \Th_{\CStrMonCat}$ is the equational extension obtained by marking the associator and unitors in $\Th_{\CMonCat}$.

  \item
    We can also consider the extension $\Th_{\ECat} \to {(\Th_{\ECat})}_E$ obtained by only marking the unitality and associativity laws.
    This extension can be better behaved than $\Th_{\CCat}$ for some applications.
    Indeed, ${(\Th_{\ECat})}_E$ has decidable equality, while $\Th_{\CCat}$ does not (\eg by reduction to the undecidability of the word problem for groups).
    \defiEnd{}
  \end{enumerate}
\end{exas}

We give some examples of equational extensions of type theories, obtained by marking a family of type equivalences and identifications between terms.
\begin{exas} \hfill
  \begin{enumerate}
  \item For the extensions from weak computation rules to strict computation rules, we mark the computation rules that should be made strict.
    For example, in the case of identity types, we mark the family of identifications $\J_{\beta}(P,d) : \Id(\J_{A,x}(P,d,\refl),d)$ over the context $(A,x,P,d)$.
  \item When considering the extension from inductive natural numbers to natural numbers with a strictly associative addition, we proceed in in two steps.
    First we extend the base theory $\Th_0$ by adding
    \[ \mathsf{plus} : \NatTy \to \NatTy \to \NatTy \]
    as a new primitive operation, along with some of the equalities that it satisfies, such as
    \begin{alignat*}{1}
      & \mathsf{plus}_{0} : \forall x \to \mathsf{plus}(0,x) \simeq x, \\
      & \mathsf{plus}_{1} : \forall x \to \mathsf{plus}(x,0) \simeq x, \\
      & \mathsf{plus}_{2} : \forall x\ y\ z \to \mathsf{plus}(\mathsf{plus}(x,y)), z) \simeq \mathsf{plus}(x, \mathsf{plus}(y,z)).
    \end{alignat*}
    The operation $\mathsf{plus}$ is homotopic to the usual inductively defined addition, but not strictly equal to it.
    The weak type theory $\Th$ is then this extended theory.
    Because $\Th$ can be seen as the extension of $\Th_0$ by new elements of contractible sorts, it should be a conservative extension of $\Th_0$ (see~\cref{conj:contractible_extension}).

    As a second step, we consider the equational extension of $\Th$ obtained by marking the identifications $\mathsf{plus}_{0}$, $\mathsf{plus}_{1}$, $\mathsf{plus}_{2}$, etc.
    Thus the extended theory $\Th_E$ includes the strict equalities $\mathsf{plus}(0,x) = x$, $\mathsf{plus}(x,0) = x$, $\mathsf{plus}(\mathsf{plus}(x,y),z) = \mathsf{plus}(x,\mathsf{plus}(y,z))$, etc.
    It also includes the strict equalities $\mathsf{plus}_{0} = \refl$, $\mathsf{plus}_{1} = \refl$, $\mathsf{plus}_{2} = \refl$, etc.

  \item
    Similarly, we can consider the extension of type theory by a composition operation on paths satisfying strictly the laws of a $1$-groupoid.
    We start by extending the base theory $\Th_0$ by new elements
    \begin{alignat*}{1}
      & \mathsf{trans} : (x \simeq y) \to (y \simeq z) \to (x \simeq z), \\
      & \mathsf{inv} : (x \simeq y) \to (y \simeq x), \\
      & \mathsf{idl} : \mathsf{trans}(\refl, p) \simeq p, \\
      & \mathsf{inv}\text{-}\refl : \mathsf{inv}(\refl) \simeq \refl, \\
      & \mathsf{idr} : \mathsf{trans}(p, \refl) \simeq p, \\
      & \mathsf{assoc} : \mathsf{trans}(\mathsf{trans}(p, q), r) \simeq \mathsf{trans}(p, \mathsf{trans}(q,r)), \\
      & \mathsf{inv}\text{-}\mathsf{l} : \mathsf{trans}(\mathsf{inv}(p), p) \simeq \refl, \\
      & \mathsf{inv}\text{-}\mathsf{r} : \mathsf{trans}(p, \mathsf{inv}(p)) \simeq \refl,
    \end{alignat*}
    along with some higher-dimensional identifications that are needed to ensure that this is an extension by new elements of \emph{contractible} sorts.
    We can see $\mathsf{idl}$ and $\mathsf{inv}\text{-}\refl$ as identifications specifying $\mathsf{trans}$ and $\mathsf{inv}$ using the universal property of $\Id$-types, so we need additional identifications specifying $\mathsf{idr}$, $\mathsf{assoc}$, $\mathsf{inv}\text{-}\mathsf{l}$ and $\mathsf{inv}\text{-}\mathsf{r}$:
    \begin{alignat*}{1}
      & - : \mathsf{idr}(\refl) \simeq \mathsf{idl}(\refl), \\
      & - : \mathsf{assoc}(p,q,\refl) \simeq \dots, \\
      & - : \mathsf{inv}\text{-}\mathsf{l}(\refl) \simeq \dots, \\
      & - : \mathsf{inv}\text{-}\mathsf{r}(\refl) \simeq \dots,
    \end{alignat*}
    where the right-hand sides should be the usual proofs of these laws.
    
    We obtain a theory $\Th$ that should be a conservative extension of $\Th_0$.

    Secondly, we consider the equational extension of $\Th$ obtained by marking the identifications $\mathsf{idl}$, $\mathsf{inv}\text{-}\refl$, \etc, as well as the additional higher-dimensional identifications.
  \item To consider the extension of a type theory with a new universe of strict propositions, we would also perform two steps.
    As a first step, we introduce a new constant type $\mathsf{SProp}$, along with an equivalence $\mathsf{SProp} \simeq \mathsf{HProp}$ with the type $\mathsf{HProp}$ of propositions.
    We write $F : \mathsf{SProp} \to \mathsf{HProp}$ for the associated transport function.

    Secondly, we mark the family of identifications
    \[ (A : \mathsf{SProp}) (a, b : F(A)) \to \Id(a,b). \]

    In the resulting strict type theory, the only way to obtain closed elements of $\mathsf{SProp}$ is to use the inverse of the equivalence $F$ to replace elements of $\mathsf{HProp}$ by elements in $\mathsf{SProp}$.

    Note that the equational extension that marks instead the family of identifications
    \[ (A : \mathsf{HProp}) (a, b : A) \to \Id(a, b) \]
    is not a conservative extension in the absence of UIP.
    Indeed, as remarked in \cite{SProp}, if all propositions are strict propositions, then UIP holds.

  \item So far we have only marked identification between terms.
    We can also mark equivalences between types.

    For example, let $\Th$ be a type theory with a weak Tarski universe closed under identity types.
    This means that $\Th$ extends $\Th_\Id$ with:
    \begin{alignat*}{1}
      & \CU : \Ty, \\
      & \El : \CU \to \Ty, \\
      & \widehat{\Id} : (A : \CU) \to \El(A) \to \El(A) \to \CU, \\
      & \alpha : \forall A\ x\ y \to \El(\widehat{\Id}(A,x,y)) \simeq \Id(\El(A),x,y).
    \end{alignat*}

    We then mark the equivalence $\alpha$ over the context $(A,x,y)$.
    The resulting theory $\Th_E$ then has a strict Tarski universe.

    Note that we have assumed that the same definition of equivalence (say relational equivalences) is used in the specification of $\alpha$ and in the specification of the homotopy relations of $\Th$ in~\cref{exa:homotopy_relations_id}.

    If $\alpha$ were instead a family of say half-adjoint equivalences, we would need to construct the extension in two steps similarly to the previous examples.
    In the first step, we add a new family $\alpha'$ of relational equivalences, along with homotopies between $\alpha$ and $\alpha'$.
    In the second step, we consider the equational extension obtained by marking the equivalences $\alpha'$.
    
  \item As a last example, we can also mark the universal identification $p$ over the generic context $(A:\Ty,x:A,y:A,p:\Id(x,y))$.
    The corresponding strict type theory then includes the equality reflection rule.
    Equivalently, one might mark all identifications over arbitrary contexts.
    \defiEnd{}
  \end{enumerate}
\end{exas}

\subsection{Morita equivalences}

We now define the notion of Morita equivalence between SOGATs.
In general, two theories are Morita equivalent when they have equivalent categories of models.
Accordingly, two SOGATs equipped with homotopy relations should be Morita equivalent when they have equivalent $\infty$-categories of models (provided that these $\infty$-categories exist).
Isaev~\cite{IsaevMorita} has defined a notion of Morita equivalence between type theories and shown that two theories are Morita equivalent if and only if their categories of models are Quillen equivalent (provided that the model or semi-model structures on the categories of models exist).
We give a simpler definition at the level of the classifying $(\reppre)$-CwFs.

Let $\Th_1 \to \Th_2$ be a morphism of SOGATs, where $\Th_1$ is equipped with homotopy relations.

\begin{defi}\label{def:theq}
  We say that the morphism $\Th_1 \to \Th_2$ is a (split/non-split) Morita equivalence if the map $\Th_1 \to \Th_2$ is a (split/non-split) weak equivalence in $\CMod_{\Th_1}$.
  \defiEnd
\end{defi}

We have an adjunction between the categories $\CMod_{\Th_1}$ and $\CMod_{\Th_2}$ of models of $\Th_1$ and $\Th_{2}$.
\begin{center}\(\begin{tikzcd}
    \CMod_{\Th_1} \ar[r, swap, bend right, "L"] \ar[r, phantom, "\top"] & \CMod_{\Th_2} \ar[l, swap, bend right, "R"]
  \end{tikzcd}\)\end{center}
When $\Th_2$ is an equational extension of $\Th_1$, the right adjoint $R$ is the forgetful functor that forgets that a model satisfies the additional equations of $\Th_2$.

We have the following characterizations of Morita equivalences.
Condition (\ref{itm:characterization_morita_eqv_1}) is Isaev's definition of Morita equivalence~\cite{IsaevMorita}.
The full proof is left to future work, but it ought to follow from the same ideas and methods as the proof of~\cref{clm:ext_univ_model_structure}.

\begin{conj}\label{clm:characterization_morita_eqv}
  Assume that $\Th_1$ satisfies external univalence.
  Then the following conditions are equivalent:
  \begin{enumerate}
  \item\label{itm:characterization_morita_eqv_1} For every cofibrant model $\CC$ of $\Th_1$, the unit $\eta_\CC : \CC \to R(L(\CC))$ is a weak equivalence in $\CMod_{\Th_1}$.
  \item\label{itm:characterization_morita_eqv_2} For every freely generated model $\CC$ of $\Th_1$, the unit $\eta_\CC : \CC \to R(L(\CC))$ is a weak equivalence in $\CMod_{\Th_1}$.
  \item\label{itm:characterization_morita_eqv_3} For every finitely generated model $\CC$ of $\Th_1$, the unit $\eta_\CC : \CC \to R(L(\CC))$ is a weak equivalence in $\CMod_{\Th_1}$.
  \item\label{itm:characterization_morita_eqv_4} The SOGAT morphism $\Th_1 \to \Th_2$ is a Morita equivalence.
  \end{enumerate}
\end{conj}
\begin{proof}[Proof sketch]\label{prf:characterization_morita_eqv}
  The fact that $\Th_1$ satisfies external univalence is needed to ensure that weak equivalences in $\CMod_{\Th_1}$ satisfy some closure conditions (some directions of the $2$-out-of-$3$ property and closure under retracts).
  
  The equivalence between (\ref{itm:characterization_morita_eqv_1}) and (\ref{itm:characterization_morita_eqv_2}) follows from the fact that every cofibrant model is a retract of a freely generated model.

  The equivalence between (\ref{itm:characterization_morita_eqv_2}) and (\ref{itm:characterization_morita_eqv_3}) follows from the fact that every $I$-cellular model is a filtered colimit of finitely generated models (because $\Th$ is a finitary theory).

  The equivalence between (\ref{itm:characterization_morita_eqv_3}) and (\ref{itm:characterization_morita_eqv_4}) follows from~\cref{clm:ff_sign}.
\end{proof}

\subsection{Extensions by contractible sorts}

In our examples of equational extensions, we have often considered compositions $\Th_0 \to \Th \to \Th_E$ of extensions, where $\Th \to \Th_E$ is an equational extension and $\Th_0 \to \Th$ is the extension of $\Th_0$ by new elements of contractible sorts.

The main focus of this paper is the conservativity of the extension $\Th \to \Th_E$, but we also need a way to prove that $\Th_0 \to \Th$ is a conservative extension.
In general, we would need a proof of the following conjecture.
\begin{conj}\label{conj:contractible_extension}
  Let $\Th$ be a cofibrant SOGAT that satisfies external univalence.
  Let $(\Gamma \vdash Y\ \type) \in \Th$ be a sort of $\Th$ that is contractible (with respect to the weakly stable identity types arising from external univalence).
  Then the extension $\Th \to \Th[\gamma:\Gamma \vdash \bm{y}(\gamma) : Y(\gamma)]$ is a Morita equivalence, and $\Th[\gamma:\Gamma \vdash \bm{y}(\gamma) : Y(\gamma)]$ also satifies external univalence (with the same homotopy relations as $\Th$).
  \defiEnd{}
\end{conj}

Proving~\cref{conj:contractible_extension} or any special case would require some more tools (including~\cref{clm:ext_univ_model_structure}), so we leave any proof to future work.

It should not too difficult to prove the special case of~\cref{conj:contractible_extension} for \emph{closed} contractible sorts, using the fact that $\Th[\bm{y} : Y]$ is isomorphic to the contextual slice $(\Th \sslash 1_\Th.Y)$, that is the contextual core (in $\CCwf_{\reppre}$) of the slice model $(\Th / 1_\Th.Y)$.

The general case is more complicated.
However, if we consider $(\repinfty)$-CwFs instead of SOGATs with external univalence (which are $(\reppre,\Id_{ws})$-CwFs), then the analogous result follows from~\cref{clm:ext_univ_model_structure}, applied to the SOGAT $\Th_{\repinfty}$ that classifies $(\repinfty)$-CwFs.
This means that~\cref{conj:contractible_extension} should follow from a coherence theorem providing a way to strictify $(\reppre,\Id_{ws})$-CwFs into $(\repinfty)$-CwFs.
This strictification is likely only possible for $(\reppre,\Id_{ws})$-CwFs which satisfy some cofibrancy condition, \eg which are cofibrant in $\CCwf_{\reppre}$.


\section{Partial saturation}\label{sec:partial_saturation}

This section introduces partially saturated $(\repinfty)$-CwFs with respect to a family $E$ of homotopies, and the initial such CwF $\Th^\infty_{\Ehat}$.
The main theorem of the paper reduces the conservativity of the equational extension $\Th \to \Th_E$ to properties of $\Th^\infty_{\Ehat}$.

\subsection{Definitions}

Let $\CC$ be a $(\repinfty)$-CwF equipped with an internal model of $\Th$.
For every generating sort $\iS \in \GenTy_\Th$, we have a comparison map going from outer identifications ($\simeq$) to homotopies ($\sim$), defined by sending the outer reflexivity $\refl$ to the constant homotopy $\hrefl$.
\begin{alignat*}{1}
  & \idToHpty_\iS : (x \simeq_{\iS(\sigma)} y) \to (x \sim_{\iS(\sigma)} y), \\
  & \idToHpty_\iS(\refl) \simeq \hrefl_{\iS}.
\end{alignat*}

The CwF $\CC$ is \emph{saturated} in the sense of~\cref{def:univalent_internal_model} exactly when these maps are equivalences for every generating sort.
Accordingly, we will say that $\CC$ is partially saturated with respect to a family of homotopies when the restriction of $\idToHpty_\iS$ to these homotopies are all equivalences.

\begin{defi}
  Let $\CC$ be a $(\repinfty)$-CwF equipped with an internal model of $\Th$, \ie a $(\reppre)$-CwF morphism $\Th \to \CC$.
  We say that $\CC$ is \defemph{partially saturated} with respect to $E$ if for every homotopy $(\gamma:\Gamma \vdash p(\gamma) : x(\gamma) \sim_{\iS(\sigma)} y(\gamma)) \in E$, the following contractibility condition holds:
  \[ \gamma:\Gamma \vdash \isContr((q : x(\gamma) \simeq_{\iS(\sigma)} y(\gamma)) \times (\idToHpty_\iS(q(\gamma)) \simeq p(\gamma))). \]
  In that case, we have the following center of contraction:
  \begin{alignat*}{1}
    & \gamma:\Gamma \vdash \widehat{p}(\gamma) : x(\gamma) \simeq_{\iS(\sigma)} y(\gamma), \\
    & \gamma:\Gamma \vdash \widetilde{p}(\gamma) : \idToHpty_\iS(\widehat{p}(\gamma)) \simeq p(\gamma).
      \tag*{\defiEnd{}}
  \end{alignat*}
\end{defi}

In particular, $\CC$ is saturated exactly when it is partially saturated with respect to all generic homotopies:
\[ \{ (\sigma:\partial \iS, x : \iS(\sigma), y : \iS(\sigma), p : x \sim_{\iS(\sigma)} y \vdash p : x \sim_{\iS(\sigma)} y) \mid \iS \in \GenTy_\Th \}, \]
or equivalently with respect to all homotopies.

We write $\Th^\infty_{\Ehat}$ for the initial object among $(\repinfty)$-CwFs equipped with an internal model of $\Th$ that is partially saturated with respect to $E$.
The superscript $-^{\infty}$ indicates that $\Th^\infty_{\Ehat}$ is morally an $\infty$-category, up to the conjectured correspondence between $(\repinfty)$-CwF and representable map $\infty$-categories.
The subscript $-_{\Ehat}$ indicates that $\Th^\infty_{\Ehat}$ includes lifts of every marked homotopy in $E$.

We write $\Th^1_{E}$ for the initial object among $(\repO)$-CwFs equipped with an internal model of $\Th_{E}$.
The superscript $-^{1}$ now indicates that $\Th^1_E$ is morally a representable map $1$-category.
Note that $\Th^1_{E}$ can be seen as a $(\repinfty)$-CwFs equipped with an internal model of $\Th$.
It is then saturated with respect to $E$, because the equations of $E$ hold definitionally in $\Th_E$.
By the universal property of $\Th^\infty_{\Ehat}$, we therefore have a morphism $\Th^\infty_{\Ehat} \to \Th^1_{E}$, such that the following square commutes:
\[ \begin{tikzcd}
    \Th
    \ar[r]
    \ar[d]
    & \Th_{E}
    \ar[d]
    \\ \Th^\infty_{\Ehat}
    \ar[r]
    & \Th^1_{E} \rlap{\ .}
  \end{tikzcd} \]

\subsection{Properties of the vertical maps}

We now discuss the properties of the two vertical maps in the above square.

Under some assumptions, these two morphisms should be seen as embeddings; at the level of categories or $\infty$-categories they would correspond to fully faithful functors or $\infty$-functors.

\begin{prop}\label{prop:the_the1_embedding}
  The $(\reppre)$-CwF morphism $\Th_{E} \to \Th^1_{E}$ is bijective on terms.
  \qed
\end{prop}
We omit the proof; it follows from somewhat standard gluing constructions or logical relation interpretations, very similar to the adequacy proof of~\cite{SketchingAndAdequacy} or to the full proof of conservativity of two-level type theory of~\cite{KovacsStagedCompilation}.
From a more syntactic point of view, it is a canonicity proof, where the canonical form of a term $(F(\Gamma) \vdash a : F(A)) \in \Th_E^1$ is a term $(\Gamma \vdash a_0 : A) \in \Th_E$ such that $F(a_0) = a$.

Alternatively, we don't need to use $\Th^1_E$, because we never need to use its universal property in this paper.
We can use the presheaf category $\widehat{\Th_E}$ instead, which is a $(\repO)$-CwF.
The Yoneda embedding $\yo : \Th_E \to \widehat{\Th_E}$ is a pseudo-morphism of $(\reppre)$-CwFs, meaning that the terminal object, context extensions and type-formers are only preserved up to isomorphism.
The action of $\yo$ on terms is bijective.
Relying on the fact that $\Th_E$ is $\{I^\ty,I^{\repty},I^\tm,E^\tm\}$-cellular, we can replace the Yoneda embedding by a strict morphism $y : \Th_E \to \widehat{\Th_E}$ along with a natural isomorphism $\alpha : y \cong \yo$.
Informally, the contexts and types of $\Th_E$ are freely closed under the operations of a $(\reppre)$-CwF, so we can redefine the actions of $y$ on contexts and types by induction over them.
Formally, we construct the iso-gluing $\CG(\yo)$ of the Yoneda embedding, equip it with the structure of a $(\reppre)$-CwF, and show that the projection $\CG(\yo) \to \Th_E$ has the right lifting property with respect to $\{I^\ty,I^{\repty},I^\tm,E^\tm\}$.
We then obtain a section $\Th_E \to \CG(\yo)$, which can be decomposed into $y$ and $\alpha$.

We only conjecture the analogous result for the other vertical map.
\begin{conj}\label{conj:embedding_from_ext_univ}
  If $\Th$ satisfies external univalence and is cofibrant in $\CCwf_{\reppre}$, then $\Th \to \Th^\infty_{\Ehat}$ is a weak equivalence in $\CMod_\Th$.
  \defiEnd{}
\end{conj}
This conjecture is discussed in more details in~\cref{sec:discussion}.


\subsection{Applications to conservativity}

We now explain how the partially saturated model $\Th_{\Ehat}^\infty$ can be used to establish that $\Th \to \Th_E$ is a Morita equivalence.
The idea is that $\Th_{\Ehat}^\infty$ is seen as an $\infty$-congruence over $\Th$ (although we don't formally define $\infty$-congruences).

An ordinary congruence consists of equivalence relations over the terms $\Th$.
Equivalently, a congruence is an embedding of a set-valued model into a setoid-valued model.
Similarly, a $\infty$-congruence should be an embedding of a set-valued model into a space-valued model, or a model valued in $\infty$-groupoids.

Here we adopt the point of view that an $\infty$-groupoid is a type in a model of type theory with identity types: the point of the $\infty$-groupoid are the terms of that type, and the higher morphisms are the terms of the iterated identity types.
All higher composition operations can be derived from the identity type eliminator.
The identity types of $\Th_{\Ehat}^\infty$ equip its terms with the structure of $\infty$-groupoids.
The higher structure of these $\infty$-groupoids is freely generated by lifts of all homotopies in $E$.

Given a $1$-congruence, it is natural to consider its quotient.
In order to prove the conservativity of the extension $\Th \to \Th_E$, we want consider the quotient of $\Th_{\Ehat}^\infty$.
To ensure the existence of the quotient, we need to assume that the $\infty$-congruence $\Th_{\Ehat}^\infty$ can actually be restricted to a $1$-congruence.

We first need to define the notion of $0$-truncated $(\lexinfty)$-CwF.
\begin{defi}
  A $(\lexinfty)$-CwF $\CC$ is said to be (merely) $0$-truncated if for every self-identification $(\gamma:\Gamma \vdash p(\gamma) : x(\gamma) \simeq_{A(\gamma)} x(\gamma))$ in $\CC$, there (merely) exists
  an identification $(\gamma:\Gamma \vdash \mathsf{uip}_p(\gamma) : p(\gamma) \simeq \refl)$ in $\CC$.
  \defiEnd{}
\end{defi}

\begin{defi}
  A $(\lexinfty)$-CwF $\CC$ is said to be (merely) $0$-truncated relatively to a functor $F : \CA \to \CC$ if for every self-identification $(\gamma:F(\Gamma) \vdash p(\gamma) : x(\gamma) \simeq_{A(\gamma)} x(\gamma))$ in $\CC$, there (merely) exists
  an identification $(\gamma:F(\Gamma) \vdash \mathsf{uip}_p(\gamma) : p(\gamma) \simeq \refl)$ in $\CC$.
  \defiEnd{}
\end{defi}

Our main result, which will be eventually proven in~\cref{sec:retracts_quotients}, is then the following:
\begin{thm}\label{thm:main_theorem_cor}
  If $\Th \to \Th^\infty_{\Ehat}$ is a weak equivalence in $\CMod_\Th$ and $\Th^\infty_{\Ehat}$ is merely $0$-truncated relatively to $\Th \to \Th^\infty_{\Ehat}$, then $\Th \to \Th_E$ is a Morita equivalence.
\end{thm}

If $\Th^\infty_{\Ehat} \to \Th^1_E$ were a weak equivalence in $\CMod_\Th$, we could conclude using $2$-out-of-$3$ for weak equivalences.
However, $\Th^\infty_{\Ehat} \to \Th^1_E$ is almost never a weak equivalence in $\CMod_\Th$ because of contexts that are not in the image of $\Th \to \Th^\infty_{\Ehat}$, such as $\Gamma \triangleq (\sigma:\partial \iS,x:\iS(\sigma),p : x \simeq_{\iS(\sigma)} x)$.
Indeed, $\idToHpty_\iS(p) \sim \hrefl$ is provable in $\Th^1_{E}$, but cannot usually be proven in $\Th^\infty_{\Ehat}$. 
But $\Th^\infty_{\Ehat} \to \Th^1_E$ doesn't need to be a full weak equivalence, it suffices that it satisfies the weak lifting conditions only over contexts that are in the image of $\Th \to \Th^\infty_{\Ehat}$.

We need a way to avoid contexts that are not in the image of $\Th \to \Th^\infty_{\Ehat}$.
We accomplish this by computing the factorization (in $\CCwf_{\reppre}$) of $\Th \to \Th_{\Ehat}$ as a $\{I^\tm,E^\tm\}$-cofibration $\Th \to {\Th_{\Ehat}}$ followed by a $\{I^\tm,E^\tm\}$-fibration ${\Th_{\Ehat}} \to \Th^\infty_{\Ehat}$.
Here $\{I^\tm,E^\tm\}$-fibrations are the right class of maps in the orthogonal factorization system generated by $I^\tm$, that is maps whose action on terms is bijective.

\begin{con}
  The $(\reppre)$-CwF $\Th_{\Ehat}$ is the unique (up to isomorphism) $(\reppre)$-CwF fitting in a factorization
  \[ \Th \xrightarrow{} \Th_{\Ehat} \xrightarrow{} \Th_{\Ehat}^\infty \]
  where $\Th \to \Th_{\Ehat}$ is a $\{I^\tm,E^\tm\}$-cofibration and $\Th_{\Ehat}^\infty$ is a $\{I^\tm,E^\tm\}$-fibration.
  \defiEnd{}
\end{con}

The CwF ${\Th_{\Ehat}}$ then plays the role of the image of $\Th \to \Th^\infty_{\Ehat}$.
It can also be identified with the subCwF of $\Th_{\Ehat}^{\infty}$ spanned by the context and types that do not contain any identity types (because they are freely generated by the type formers of $(\reppre)$-CwFs).

The subscript $-_{\Ehat}$ still indicates that $\Th_{\Ehat}$ includes lifts of all marked homotopies in $E$, or more precisely all transports and higher-dimensional transports over these lifts.
The absence of the superscript $-^{\infty}$ indicates that $\Th_{\Ehat}$ can no longer be seen as an $\infty$-category.

\begin{lem}
  The map $\Th \to {\Th_{\Ehat}}$ is a $\{I^\ty,E^\ty,I^{\repty},E^\repty,I^\tm,E^\tm\}$-cofibration.
\end{lem}
\begin{proof}
  The map $\Th \to {\Th_{\Ehat}}$ is a $\{I^\tm, E^\tm\}$-cofibration by definition and
  \[
    \{I^\tm, E^\tm\} \subseteq \{I^\ty, E^\ty, I^{\repty}, E^\repty, I^\tm, E^\tm\}.
    \tag*{\qedhere}
  \]
\end{proof}

\begin{lem}\label{lem:lt_tinfty_is_local_iso}
  The map ${\Th_{\Ehat}} \to \Th_{\Ehat}^\infty$ is a contextual isomorphism in $\CMod_\Th$.
\end{lem}
\begin{proof}
  Since that map is bijective on terms in $\CCwf_{\reppre}$, it is in particular bijective on the elements of any sort, \ie a contextual isomorphism.
\end{proof}

\begin{lem}\label{lem:weq_ehat}
  If $\Th \to \Th^\infty_{\Ehat}$ is a weak equivalence in $\CMod_\Th$, then $\Th \to {\Th_{\Ehat}}$ is also a weak equivalence in $\CMod_\Th$.
\end{lem}
\begin{proof}
  By $2$-out-of-$3$, since ${\Th_{\Ehat}} \to \Th^\infty_{\Ehat}$ is a weak equivalence in $\CMod_\Th$ as a consequence of~\cref{lem:lt_tinfty_is_local_iso}.
  This direction of $2$-out-of-$3$ in $\CMod_\Th$ was proven in~\cref{lem:modt_two_out_of_three_2}.
\end{proof}

The lifting property of the $\{I^\ty,E^\ty,I^{\repty},E^\repty,I^\tm,E^\tm\}$-cofibration $\Th \to {\Th_{\Ehat}}$ against the contextual isomorphism $\Th_E \to \Th_E^1$ provides a morphism ${\Th_{\Ehat}} \to \Th_E$ of $(\reppre)$-CwFs.
The following diagram commutes:
\[ \begin{tikzcd}
    \Th
    \ar[rr]
    \ar[dd]
    \ar[rd]
    && \Th_{E}
    \ar[dd, two heads, "{\sim}", "L"']
    \\
    &
    \Th_{\Ehat}
    \ar[ld, two heads, "{\sim}", "K"']
    \ar[ru, "F"]
    &
    \\
    \Th^\infty_{\Ehat}
    \ar[rr, "G"]
    && \Th^1_{E} \rlap{\ .}
  \end{tikzcd} \]
The remainder of the paper develops tools used to prove that $\Th_{\Ehat} \to \Th_E$ is a non-split trivial fibrations whenever $\Th_{\Ehat}^\infty$ is merely $0$-truncated relatively to $\Th \to \Th_{\Ehat}^\infty$.


\section{Quotients of models of first-order GATs}\label{sec:fibrant_congruences}

In this section we discuss congruences over models of (first-order) GATs and quotients of such congruences.
In this paper, we only need applications to the GAT of $(\reppre)$-CwFs, but it is easier to present them in full generality.

For ordinary algebraic theories, congruences and quotients are very well-behaved.
A congruence simply consists of equivalence relations for every sort, preserving all operations in the theory.
The quotient of a model is computed by taking the quotient of every sort by the specified equivalence relation.

For generalized algebraic theories, matters are more complicated, because of the presence of dependent sorts.
We still have a general notion of congruence, but quotients can no longer be computed sortwise.
In other words, the quotient inclusion is not necessarily a \emph{trivial fibration}.

For example, we can consider the quotient of the category
$\begin{tikzcd}[column sep=20pt]
  a \ar[r,"f"] & b & c \ar[r,"g"] & d
\end{tikzcd}$
by the congruence that identifies $b$ and $c$.
In the quotient, it is possible to compose $f$ and $g$, but there is no morphism corresponding to their composition in the original category.
The quotient inclusion functor is not full.

However, by adding additional fibration condition to the definition of congruence, it becomes possible to compute the quotient pointwise.
The quotient inclusion of any ``fibrant'' congruence is then a non-split trivial fibration.
Furthermore, every non-split trivial fibration arises (up to isomorphism) as the quotient of a fibrant congruence.

\subsection{Functorial semantics of GATs}

We start by briefly recalling the functorial semantics of GATs in $\CCwf_{\lexpre}$, in the same style as for the semantics of SOGATs developed in~\cref{sec:sogats}.
\begin{defi}
  A \defemph{generalized algebraic theory} (\defemph{GAT}) is a $\{I^\ty,I^\tm,E^\tm\}$-cellular $\lexpre$-CwF.
  \defiEnd{}
\end{defi}

Any SOGAT without representable sorts is also a GAT (this can be proven by induction over the $\{I^\ty,I^\tm,E^\tm\}$-cellular presentation of the SOGAT).
This includes the GATs of categories (\cref{exa:gat_cat}), monoidal categories (\cref{exa:gat_moncat}), \etc

\begin{exa}
  The GAT $\Th_{\CCwf}$ of CwFs extends the GAT of categories with two new sorts
  \begin{alignat*}{1}
    & \iTy : \iOb \to \MT_{\CCwf}, \\
    & \iTm : (\Gamma : \iOb) \to \iTy(\Gamma) \to \MT_{\CCwf}, 
  \end{alignat*}
  six operations
  \begin{alignat*}{1}
    & \_[\_] : \iTy(\Gamma) \to \iHom(\Delta,\Gamma) \to \iTy(\Delta), \\
    & \_[\_] : \iTm(\Gamma,A) \to (\gamma : \iHom(\Delta,\Gamma)) \to \iTy(\Delta,A[\gamma]), \\
    & \_.\_ : (\Gamma : \iOb) \to \iTy(\Gamma) \to \iOb, \\
    & \bm{p} : \iHom(\Gamma. A, \Gamma), \\
    & \bm{q} : \iTm(\Gamma. A, A[\bm{p}]), \\
    & \angles{\_,\_} : (\gamma : \iHom(\Delta,\Gamma)) \to \iTm(\Delta, A[\gamma]) \to \iHom(\Delta, \Gamma.A)
  \end{alignat*}
  and some equations expressing the functoriality and naturality of these operations.

  Similarly, there is a GAT $\Th_{\CCwf_{\reppre}}$ of $(\reppre)$-CwFs, which extends $\Th_\CCwf$ by a new sort of representable types, a substitution operation for representable types, an operation for every type-theoretic operation of a $(\reppre)$-CwF, and equations expressing the functoriality and naturality of the operations.
  \defiEnd{}
\end{exa}

We now fix a GAT $\Th$ for most of the section.
A closed sort of $\Th$ is an object $X \in \Th$ (or equivalently a closed type).
A dependent sort of $\Th$ is a type $(X \vdash Y\ \type) \in \Th$.

\begin{defi}
  A model of $\Th$ is a $\lexpre$-CwF morphism $\Th \to \CSet$.
  A morphism between two models $\CM,\CN : \Th \to \CSet$ is a natural transformation $\CM \Ra \CN$.
  \defiEnd{}
\end{defi}

Given a closed sort $X \in \Th$ and morphism $F : \CM \to \CN$, the component
\[ F_X : \CM_X \to \CN_X \]
of the natural transformation is the action of $F$ on elements of the sort $X$.

For any dependent sort $(X \vdash Y\ \type) \in \Th$ and morphism $F : \CM \to \CN$, there is an action
\[ F_Y : \CM_Y(x) \to \CN_Y(F_X(x)) \]
of $F$ on elements of $Y$, uniquely specified by the fact that 
\[ F_{X.Y} : \CM_{X.Y} \to \CN_{X.Y} \]
should send $(x,y)$ to $(F_X(x),F_Y(y))$.

\begin{defi}\label{def:tfib_gat}
  A morphism $F : \CM \to \CN$ in $\CMod_\Th$ is a \defemph{trivial fibration} if for every dependent sort $(X \vdash Y\ \type) \in \Th$, the map
  \[ F_Y : \CM_Y(x) \to \CN_Y(F_X(x)) \]
  is surjective.
  \defiEnd{}
\end{defi}
As in earlier definitions of weak equivalences and trivial fibrations, we actually have two notions of trivial fibrations: split trivial fibrations and non-split trivial fibrations.

\begin{rem}
  When applied to the GAT $\Th_{\CCwf}$ of CwFs, trivial fibrations as defined in~\cref{def:tfib_gat} are not exactly the same as trivial fibrations as defined in~\cref{def:tfib_cwf}, which could be called local trivial fibrations.
  Indeed, \cref{def:tfib_gat} includes additional lifting conditions for context and substitutions.
  Fortunately, any trivial fibration is a local trivial fibration, which is the direction needed for our constructions.
  \defiEnd{}
\end{rem}

\subsection{Congruences}

Given an equivalence relation $\widetilde{X}$ over a set $X$, we write $(x_1 \sim x_2) \in \widetilde{X}$ to indicate that $x_1$ and $x_2$ .
We use the same notation for dependent equivalence relations.
Given a dependent equivalence relation $\widetilde{Y}$ over $\widetilde{X}$, writing $(y_1 \sim y_2) \in \widetilde{Y}$ presupposes that $(x_1 \sim x_2) \in \widetilde{X}$ (when $y_1 : Y(x_1)$ and $y_2 : Y(x_2)$).

The $\lexpre$-CwF $\CSetoid$ of setoids is a CwF over the category of setoids in which types are dependent (or displayed) setoids, \ie families of sets equipped with a dependent equivalence relation.

\begin{defi}
  A \defemph{congruence} over a model $\CM : \Th \to \CSet$ is a $\lexpre$-CwF morphism $\widetilde{\CM} : \Th \to \CSetoid$ such that $U \circ \widetilde{\CM} = \CM$.
  \defiEnd{}
\end{defi}

Given a congruence $\widetilde{\CM}$ and a closed sort $X \in \Th$, we have an equivalence relation $\widetilde{\CM}_X$ on $\CM_X$.
For a dependent sort $(X \vdash Y\ \type) \in \Th$, we have a dependent equivalence relation $\widetilde{\CM}_Y$ over $\widetilde{\CM}_X$.
Given any morphism or term in $\Th$, the action of $\widetilde{\CM}$ on morphisms and terms implies that the corresponding operation of $\CM$ preserves the (dependent) equivalence relations.

Given two congruences $\widetilde{\CM}$ and $\overline{\CM}$ over the same model $\CM$, we write $\widetilde{\CM} \subseteq \overline{\CM}$ when $\widetilde{\CM}_X \subseteq \overline{\CM}_X$ as relations for every closed sort or dependent sort $X$.

\begin{defi}
  The \defemph{kernel} of a morphism $F : \CM \to \CN$ of models is the congruence $\ker_F : \Th \to \CSetoid$ over $\CM$ defined by:
  \[ ((x_1 \sim x_2) \in \ker_F(X)) \overset{\triangle}{\iff} F(x_1) = F(x_2) \]
  for any closed sort $X \in \Th$ and
  \[ ((y_1 \sim y_2) \in \ker_F(Y)) \overset{\triangle}{\iff} F(y_1) = F(y_2) \]
  for any dependent sort $(X \vdash Y\ \type) \in \Th$.
  \defiEnd{}
\end{defi}

\begin{defi}
  A \defemph{quotient} of a congruence $\widetilde{\CM}$ is a model $\Quot$ along with a morphism $\quot : \CM \to \Quot$ such that $\widetilde{\CM} \subseteq \ker_{\mathbf{q}}$ and satisfying the following universal property: for every other morphism $F : \CM \to \CN$ such that $\widetilde{\CM} \subseteq \ker_F$, there is a unique morphism $\widetilde{F} : \Quot \to \CN$ such that $\widetilde{F} \circ \quot = F$.
  \defiEnd{}
\end{defi}

\begin{lem}
  A quotient of a congruence always exists.
\end{lem}
\begin{proof}
  The quotient can be expressed as a coequalizer and the category of models of $\Th$ is cocomplete.
\end{proof}

\subsection{Fibrant congruences and quotients}

\begin{defi}
  A dependent equivalence relation $\widetilde{Y}$ over an equivalence relation $\widetilde{X}$ is \defemph{fibrant} if it satisfies the following lifting condition:
  for every $(x_1 \sim x_2) \in \widetilde{X}$ and $y_1 : Y(x_1)$, there merely exists some $y_2 : Y(x_2)$ such that $(y_1 \sim y_2) \in \widetilde{Y}$.
  \defiEnd{}
\end{defi}

We can restrict the types of $\CSetoid$ to fibrant dependent equivalence relations.
This provides a sub-CwF $\CSetoid_\fib \subseteq \CSetoid$ over the same base category.

\begin{defi}
  A congruence $\widetilde{\CM} : \Th \to \CSetoid$ is \defemph{fibrant} if it factors through the inclusion $\CSetoid_\fib \subseteq \CSetoid$, \ie if for every dependent sort $(X \vdash Y\ \type) \in \Th$, the dependent equivalence relation $\widetilde{\CM}_Y$ is fibrant.
  \defiEnd{}
\end{defi}

\begin{lem}
  The quotient functor $Q : \CSetoid_\fib \to \CSet$ extends to a $\lexpre$-CwF morphism.
\end{lem}
\begin{proof}
  We have to prove that $Q$ preserves context extensions.
  This boils down to the fact that $Q$ preserves pullback squares when at least one of the maps of the cospan is a setoid fibration.

  Take a pullback square in $\CSetoid$.
  \[ \begin{tikzcd}
      B
      \ar[rd, phantom, very near start, "\lrcorner"]
      \ar[r]
      \ar[d, two heads]
      & Y
      \ar[d, "p", two heads]
      \\ A
      \ar[r, "f"]
      & X\rlap{ ,}
    \end{tikzcd} \]
  where $Y \thra X$ is a setoid fibration.
  We write $[-]$ for the inclusions into the quotients of the setoids.
  Our goal is to prove that the canonical map $Q(B) \to Q(A) \times_{Q(X)} Q(Y)$ is an isomorphism.

  We define the inverse $Q(A) \times_{Q(X)} Q(Y) \to Q(B)$ using the universal properties of the quotients.
  Take a pair $([a],[y])$ in $Q(A) \times_{Q(X)} Q(Y)$.
  We have $[f(a)] = [p(y)]$, \ie $f(a) \sim p(y)$ in the setoid $X$.
  Since $p : Y \thra X$ is a fibration, we can find $y' \in Y$ such that $p(y) = f(a)$ and $y' \sim y$. We send the pair $([a],[y])$ to the element $[(a,y')]$ of $Q(B)$.
  It can be shown that element does not depend on the choices of $a$, $y$ and $y'$, thus the candidate inverse is well-defined.

  It is then straightforward to show that the defined map is an inverse.
\end{proof}

\begin{lem}\label{lem:fib_congruence_quotient}
  Given a fibrant congruence $\widetilde{\CM}$, the model $\Quot(\widetilde{\CM}) \triangleq (Q \circ \widetilde{\CM}) : \Th \to \CSet$ is the quotient of $\widetilde{\CM}$.
  The quotient inclusion $\quot : \CM \Ra \Quot(\widetilde{\CM})$ is the natural transformation $\widetilde{\CM} \cdot q$, where $q : U \Ra Q$ is the natural transformation that maps the underlying sets of setoids to their quotients.
\end{lem}
\begin{proof}
  Let $\CN$ be any model of $\Th$, along with a morphism $F : \CM \to \CN$ such that $\widetilde{\CM} \subseteq \ker_F$.
  We define a natural transformation $\widetilde{F} : \Quot(\widetilde{\CM}) \Ra \CN$.

  For any closed sort $X \in \Th$, we define $\widetilde{F}_X : Q(\widetilde{\CM}_X) \to \CN_X$ using the universal property of $Q(\widetilde{\CM}_X)$, as the unique lift of $F_X : \CM_X \to \CN_X$.

  For any morphism $f : X \to Y$, the commutation of the square
  \[ \begin{tikzcd}
      Q(\widetilde{\CM}_X)
      \ar[r, "Q(\widetilde{\CM}_f)"]
      \ar[d, "\widetilde{F}_X"]
      & Q(\widetilde{\CM}_Y)
      \ar[d, "\widetilde{F}_Y"]
      \\
      \widetilde{\CN}_X
      \ar[r, "\CN_f"]
      & \widetilde{\CN}_Y
    \end{tikzcd} \]
  follows from the universal property of the quotient $Q(\widetilde{\CM}_X)$.

  The uniqueness of $\widetilde{F}$ as a lift of $F$ follows from the uniqueness of every component $\widetilde{F}_X$ as a lift of $F_X$.
\end{proof}

\begin{cor}\label{cor:fib_cong_tfib}
  If $\widetilde{\CM}$ is a fibrant congruence, then $\quot : \CM \to \Quot(\widetilde{\CM})$ is a non-split trivial fibration in $\CMod_\Th$.
\end{cor}
\begin{proof}
  By~\cref{lem:fib_congruence_quotient}, all components of $\quot$ are quotient inclusions, which are all non-split surjections.
\end{proof}

\begin{lem}
  If $F : \CM \to \CN$ is a non-split trivial fibration in $\CMod_\Th$, then $\ker_F$ is a fibrant congruence and $\CN$ is a quotient of $\ker_F$.
\end{lem}
\begin{proof}
  We first prove the fibrancy of $\ker_F$.
  
  Let $(X \vdash Y\ \type) \in \Th$ be a dependent sort of $\Th$.
  Take $x_1,x_2 : \CM_X$ such that $F(x_1) = F(x_2)$ and $y_1 : \CM_Y(x_1)$.
  Then $F(y_1) : \CN_Y(F(x_2))$.
  Since $F$ is a trivial fibration, there exists a lift $y_2 : \CM_Y(x_2)$ such that $F(y_2) = F(y_2)$.
  This proves exactly that $\ker_F$ is fibrant.

  We can now construct the quotient $\Quot(\ker_F)$ (which could be called the coimage of $F$), and we know that the quotient inclusion is a non-split trivial fibration.
  The map $F$ factors through the quotient inclusion.
  \[ \begin{tikzcd}
      \CM
      \ar[rr, "F", two heads, "{\sim}"']
      \ar[rd, "\quot"', two heads, "{\sim}"]
      && \CN
      \\
      & \Quot(\ker_F)
      \ar[ru, "G"']
    \end{tikzcd} \]
  Since $F$ is a non-split trivial fibration, $G$ is also a non-split trivial fibration, surjective on every sort.
  To prove that $G$ is an isomorphism, it suffices to show that $G$ is injective on every sort.

  Let $(X \vdash Y\ \type) \in \Th$ be a dependent sort.
  Take $x : {\Quot(\ker_F)}_X$ and $y_1,y_2 : {\Quot(\ker_F)}_Y(x)$ such that $G(y_1) = G(y_2)$.
  Since $\quot$ is a non-split trivial fibration there exists lifts $x' : \CM_X$ and $y_1',y_2' : \CM_Y(x')$ of $x$, $y_1$ and $y_2$.
  We then have $F(y'_1) = F(y'_2)$, \ie $(y'_1 \sim y'_2) \in \ker_F$.
  Therefore, $\quot(y'_1) = \quot(y'_2)$, that is $y_1 = y_2$.
  
  This proves that the actions of $G$ on every sort are injective.
  Therefore, $G$ is an isomorphism and $\CN$ is a quotient of $\ker_F$.
\end{proof}


\section{Trivial fibrations as retracts of quotient inclusions}\label{sec:retracts_quotients}

Fix a SOGAT $\Th$ equipped with homotopy relations and an equational extension $\Th \to \Th_E$.
Recall from~\cref{sec:partial_saturation} that we have the following commutative diagram in $\CCwf_{\reppre}$:
\[ \begin{tikzcd}
    \Th
    \ar[rr]
    \ar[dd]
    \ar[rd]
    && \Th_{E}
    \ar[dd, two heads, "{\sim}", "L"']
    \\
    &
    \Th_{\Ehat}
    \ar[ld, two heads, "{\sim}", "K"']
    \ar[ru, "F"]
    &
    \\
    \Th^\infty_{\Ehat}
    \ar[rr, "G"]
    && \Th^1_{E} \rlap{\ .}
  \end{tikzcd} \]

The goal of this section is to prove the following theorem:
\begin{thm}\label{thm:main_theorem}
  Let $\Th$ be a SOGAT with homotopy relations, and let $E$ be a set of homotopies of $\Th$.
  If $\Th^\infty_{\Ehat}$ is merely $0$-truncated relatively to $K : \Th \to \Th^\infty_{\Ehat}$, then $F : \Th_{\Ehat} \to \Th_E$ is a trivial fibration in $\CCwf_{\reppre}$.
\end{thm}

We follow the structure of Hofmann's proof of the conservativity of Extensional Type Theory over Intensional Type Theory.
We will construct a quotient $\Quot(\widetilde{\Th_{\Ehat}})$ in $\CCwf_{\reppre}$ of $\Th_{\Ehat}$ by a fibrant congruence $\widetilde{\Th_{\Ehat}}$, and prove that the map $\Th_{\Ehat} \to \Th_E$ becomes a retract of the quotient inclusion $\quot : \Th_{\Ehat} \to \Quot(\widetilde{\Th_{\Ehat}})$.
The results then follows from the closure of non-split trivial fibrations under retracts.
\[ \begin{tikzcd}
    &
    \Quot(\widetilde{\Th_{\Ehat}})
    \ar[rd, "R"']
    & \\
    \Th_{\Ehat}
    \ar[rr, "F"']
    \ar[ru, "\quot"]
    &&
    \Th_E
    \ar[lu, bend right, "S"']
  \end{tikzcd} \]

\subsection{Equivalences and dependent equivalences}

For this proof we need to consider some equivalences over the contexts and types of $\Th_{\Ehat}^\infty$.
It is most convenient to work with relational equivalences, also known as one-to-one correspondences.
A relational equivalence between two types $A$ and $B$ is a type-valued relation that is functional in both directions.
The advantage of this definition is that is easy to show that the standard type structures ($\Unit$, $\Sigma$ and $\Pi$) preserve equivalences.

Let $\MC$ be a family equipped with $\Sigma$-types and weak identity types.

\begin{defi}
  Let $A_1,A_2 : \MC$ be two types.
  An \defemph{equivalence} $\alpha : A_1 \simeq A_2$ is a type-valued relation
  \[ \alpha : A_1 \to A_2 \to \MC \]
  satisfying the following two contractibility conditions
  \begin{alignat*}{1}
    & (a_1 : A_1) \to \isContr((a_2 : A_2) \times \alpha(a_1,a_2)), \\
    & (a_2 : A_2) \to \isContr((a_1 : A_1) \times \alpha(a_1,a_2)).
      \tag*{\defiEnd{}}
  \end{alignat*}
\end{defi}

The elements of $\alpha(a_1,a_2)$ will be called identifications between $a_1$ and $a_2$.

\begin{defi}
  Let $B_1 : A_1 \to \MC$ and $B_2 : A_2 \to \MC$ be two dependent types.
  A \defemph{dependent equivalence} $\beta : B_1 \simeq_\alpha B_2$ over an equivalence $\alpha : A_1 \simeq A_2$ is a family
  \[ \beta : \forall a_1\ a_2 \to \alpha(a_1,a_2) \to B_1(a_1) \simeq B_2(a_2) \]
  of equivalences.
  \defiEnd{}
\end{defi}

\begin{defi}
  Let $\alpha_1,\alpha_2 : A_1 \simeq A_2$ be two parallel equivalences.
  A \defemph{homotopy} $H : \alpha_1 \simeq \alpha_2$ between $\alpha_1$ and $\alpha_2$ is a family
  \[ H : \forall a_1\ a_2 \to \alpha_1(a_1,a_2) \simeq \alpha_2(a_1,a_2) \]
  of equivalences.
  \defiEnd{}
\end{defi}

\begin{defi}
  Let $\beta_1 : B_1 \simeq_{\alpha_1} B_2$ and $\beta_2 : B_1 \simeq_{\alpha_2} B_2$ be two parallel dependent equivalences, and $H : \alpha_1 \simeq \alpha_2$ be a homotopy between $\alpha_1$ and $\alpha_2$.
  A \defemph{dependent homotopy} $G$ between $\beta_1$ and $\beta_2$ is a family
  \begin{alignat*}{1}
    & G : \forall a_1\ a_2\ (a_{e1} : \alpha_1(a_1,a_2))\ (a_{e2} : \alpha_2(a_1,a_2))\ (a_{ee} : H(a_{e1},a_{e2})) \\
    & \quad \to \forall b_1\ b_2 \to \beta_1(a_{e1},b_1,b_2) \simeq \beta_2(a_{e2},b_1,b_2) 
  \end{alignat*}
  of equivalences.
  \defiEnd{}
\end{defi}

These notions come with various composition operations, including identity equivalences, inverse equivalences, \etc.

\subsection{Construction of the fibrant congruence}

We now assume that $\Th^\infty_{\Ehat}$ is merely $0$-truncated relatively to $K : \Th \to \Th^\infty_{\Ehat}$.

We want to construct a quotient of $\Th_{\Ehat}$ by a fibrant congruence.
This congruence should identify terms of $\Th_{\Ehat}$ when there is an identification between them in $\Th_{\Ehat}^\infty$, and it should identify types of $\Th_{\Ehat}$ when they are equivalent in $\Th_{\Ehat}^\infty$.
Since $\Th_{\Ehat}^\infty$ is $0$-truncated relatively to $K$, the choice of identifications won't matter.
However, there may be multiple parallel equivalences in $\Th_{\Ehat}^\infty$.
Because of that, we need to consider a restricted class of equivalences, built from transport over identifications.
This is analogous to the partial families $\mathsf{co}_{\Gamma,\Delta}$ and $\mathsf{ty}_{\Gamma,\Delta,\sigma,\tau}$ of propositional isomorphisms that are used in Hofmann's conservativity proof~\cite{HofmannCons}.

\begin{con}
  We define the following set-valued relations and families of equivalences by induction over the structure of contexts and types of $\Th_{\Ehat}$.
  \begin{alignat*}{1}
    & ({\approx}) : \Th_{\Ehat} \to \Th_{\Ehat} \to \SSet, \\
    & ({\approx}) : (\Gamma_1 \approx \Gamma_2) \to \Th_{\Ehat}.\Ty \to \Th_{\Ehat}.\Ty \to \SSet, \\
    & ({\approx})^{\rep} : (\Gamma_1 \approx \Gamma_2) \to \Th_{\Ehat}.\Ty_\rep \to \Th_{\Ehat}.\Ty_\rep \to \SSet, \\
    & \alpha : (\Gamma_1 \approx \Gamma_2) \to (K(\Gamma_1) \simeq K(\Gamma_2)), \\
    & \beta : (p : \Gamma_1 \approx \Gamma_2) \to (A \approx_p B) \to (K(A) \simeq_{\alpha_p} K(B)), \\
    & \beta^{\rep} : (p : \Gamma_1 \approx \Gamma_2) \to (A \approx^\rep_p B) \to (K(A) \simeq_{\alpha_p} K(B)).
  \end{alignat*}

  \begin{description}
  \item[Empty context]
    We pose $(\diamond \approx \diamond) \triangleq \{\star\}$.
    Given $p : \diamond \approx \diamond$, we let $\alpha_p$ be the identity equivalence, that is
    \[ \alpha_p(\_,\_) \triangleq \Unit. \]

  \item[Context extension]
    We pose $((\Gamma_1.A_1) \approx (\Gamma_2.A_2)) \triangleq (p : \Gamma_1 \approx \Gamma_2) \times (q : A_1 \approx_p A_2)$.
    Given $(p,q) : (\Gamma_1.A_1) \approx (\Gamma_2.A_2)$, we let $\alpha_{(p,q)}$ be the equivalence obtained by lifting the equivalences $\alpha_p$ and $\beta_q$ to extended contexts, \ie
    \[ \alpha_{(p,q)}((\gamma_1,a_1),(\gamma_2,a_2)) \triangleq (\gamma_e : \alpha_p(\gamma_1,\gamma_2)) \times (a_e : \beta_q(\gamma_e,a_1,a_2)). \]

  \item[Basic types]
    We let $(\iS(\sigma_1) \approx_p \iS(\sigma_2))$ be the set of identifications between $K(\sigma_1)$ and $K(\sigma_2)$ over the equivalence $\alpha_p$, that is an element of $(\iS(\sigma_1) \approx_p \iS(\sigma_2))$ is an element
    \[ (\gamma_1:K(\Gamma_1),\gamma_2:K(\Gamma_2),\gamma_e:\alpha_p(\gamma_1,\gamma_2) \vdash q(\gamma_e) : K(\sigma_1)(\gamma_1) \simeq_{\partial \iS} K(\sigma_2)(\gamma_2)) \in \Th_{\Ehat}^\infty \]

    Given such an identification $q$, we let $\beta_q$ be the dependent equivalence defined by transport over $q$, namely
    \[ \beta_q(\gamma_e,a_1,a_2) \triangleq \mathsf{transport}_{\iS(-)}(a_1,q(\gamma_e)) \simeq_{\iS(K(\sigma_2)(\gamma_2))} a_2. \]

  \item[$\Unit$-types]
    We pose $(\Unit \approx_p \Unit) \triangleq \{\star\}$.
    Given $q : \Unit \approx_p \Unit$, we let $\beta_{q}$ be the identity equivalence, that is
    \[ \beta_q(\gamma_e,\_,\_) \triangleq \Unit. \]

  \item[$\Sigma$-types]
    We pose $(\Sigma(A_1,B_1) \approx_p \Sigma(A_2,B_2)) \triangleq (q_A : A_1 \approx_p A_2) \times (q_B : B_1 \approx_{p,q_A} B_2)$.
    Given $q : \Sigma(A_1,B_1) \approx_p \Sigma(A_2,B_2)$, we let $\beta_q$ be the equivalence obtained by lifting the equivalences $\beta_{q_A}$ and $\beta_{q_B}$ to $\Sigma$-types, that is
    \[ \beta_{q}(\gamma_e,(a_1,a_2),(b_1,b_2)) \triangleq (a_e : \beta_{q_A}(\gamma_e,a_1,a_2)) \times (b_e : \beta_{q_B}((\gamma_e,a_e),b_1,b_2)). \]
    
  \item[$\Pi_\rep$-types]
    We pose $(\Pi_\rep(A_1,B_1) \approx_p \Pi_\rep(A_2,B_2)) \triangleq (q_A : A_1 \approx^\rep_p A_2) \times (q_B : B_1 \approx_{p,q_A} B_2)$.
    Given $q : \Pi_\rep(A_1,B_1) \approx_p \Pi_\rep(A_2,B_2)$, we let $\beta_p(q)$ be the equivalence obtained by lifting the equivalences $\beta_{q_A}$ and $\beta_{q_B}$ to $\Pi_\rep$-types, namely
    \[ \beta_{q}(\gamma_e,f_1,f_2) \triangleq (a_1 : A_1(\gamma_1))\ ((a_2,a_e) := \overrightarrow{\beta_{q_A}}(a_1)) \to \beta_{q_B}((\gamma_e,a_e),f_1(a_1),f_2(a_2)). \]

    Here $a_2 : K(A_2)(\gamma_2)$ and $a_e : \beta_{q_A}(\gamma_e,a_1,a_2)$ are defined by transporting $a_1$ over the equivalence $\beta_{q_A}$.
    It would not be possible to quantify over $(a_2,a_e)$ instead, because $a_e : \beta_{q_A}(a_1,a_2)$, but $\beta_{q_A}(a_1,a_2)$ is assumed to be representable.
    
  \item[Representable types]
    The relation ${\approx}^\rep$ and the function $\beta^\rep$ are defined analogously to ${\approx}$ and $\beta$.
    Given two representable types $A$ and $B$, we pose $A \approx_p B \triangleq A \approx^\rep_p B$.
    Given $q : A \approx_p B$, we pose $\beta_q \triangleq \beta_q^\rep$.
    
  \item[Remaining cases]
    In any other case, \ie when the two contexts or types have different shapes, they are not related by ${\approx}$, \ie we pose
    \[ (A \approx B) \triangleq \emptyset.
      \tag*{\defiEnd{}}
    \]
  \end{description}
\end{con}

We will use induction on the structure of context and types of $\Th_{\Ehat}$ in multiple proofs.
In all of those, we view basic types (of the form $\iS(\sigma)$) as the base cases of the induction, and all other cases as recursive cases of the induction.

We now prove that the relations $(\approx)$ are reflexive, transitive and symmetric up to homotopy.

\begin{lem}\label{lem:approx_refl}
  For every $\Gamma \in \Th_{\Ehat}$, there is an element $r_\Gamma : \Gamma \approx \Gamma$ and a homotopy between the equivalence $\alpha_{r_\Gamma} : K(\Gamma) \simeq K(\Gamma)$ and the identity equivalence.

  For every $(\Gamma \vdash A\ \type) \in \Th_{\Ehat}$, there is an element $r_A : A \approx_{r_\Gamma} A$ and a homotopy between the dependent equivalence $\beta_{r_A} : K(A) \simeq_{\alpha_{r_\Gamma}} K(A)$ and the identity equivalence.

  The analogous statement also holds for representable types.
\end{lem}
\begin{proof}
  By mutual induction on the structures of $\Gamma$ and $A$. In every case we have to use the fact that the corresponding construction of $\alpha$ or $\beta$ preserves identity equivalences.
\end{proof}

We can prove both symmetry and transitivity in a single step.
\begin{lem}\label{lem:approx_trans}
  For every $\Gamma_1,\Gamma_2,\Gamma_3 : \Th_{\Ehat}$ and elements $p_1 : \Gamma_1 \approx \Gamma_2$ and $p_2 : \Gamma_3 \approx \Gamma_2$, there is an element $t(p_1,p_2) : \Gamma_1 \approx \Gamma_3$ and a homotopy between $\alpha(t(p_1,p_2)) : K(\Gamma_1) \simeq K(\Gamma_3)$ and the composed equivalence $\alpha(p_1) \cdot {\alpha(p_2)}^{-1}$.

  Given types $(\Gamma_1 \vdash A_1\ \type) \in \Th$, $(\Gamma_2 \vdash A_2\ \type) \in \Th$ and $(\Gamma_3 \vdash A_3\ \type) \in \Th$, and elements $q_1 : A_1 \approx_{p_1} A_2$ and $q_2 : A_3 \approx_{p_2} A_2$, there is an element $t(q_1,q_2) : A_1 \approx_{t(p_1,p_2)} A_3$ and a homotopy between $\beta(t(q_1,q_2))$ and the composed equivalence $\beta(q_1) \cdot {\beta(q_2)}^{-1}$.

  The analogous statement also holds for representable types.
\end{lem}
\begin{proof}
  By mutual induction on the structures of the contexts and types.
  In every case we have to use the fact that the corresponding construction of $\alpha$ or $\beta$ preserves inverse equivalences and compositions of equivalences.
\end{proof}
Symmetry for $p : \Gamma \approx \Delta$ can then be recovered as $t(r_\Gamma,p)$, and transitivity for $p : \Gamma \approx \Delta$ and $q : \Delta \approx \Theta$ can be reconstructed as $t(p,t(r_\Delta,q))$.

\begin{lem}\label{lem:approx_uip}
  Given two elements $p_1,p_2 : \Gamma_1 \approx \Gamma_2$, there merely exists a homotopy between the equivalences $\alpha_{p_1}$ and $\alpha_{p_2}$.

  Given two elements $q_1 : A_1 \approx_{p_1} A_2$ and $q_2 : A_1 \approx_{p_2} A_2$, there merely exists a homotopy $h$ between the equivalences $\alpha_{p_1}$ and $\alpha_{p_2}$ and a dependent homotopy between $\beta_{q_1}$ and $\beta_{q_2}$, over the homotopy $h$.
\end{lem}
\begin{proof}
  By induction on the structure of the contexts and types.
  In the base case ($\iS(\sigma_1) \approx \iS(\sigma_2)$), we use the fact that $\CI_E^\infty$ is merely $0$-truncated relatively to $K$.
  In other cases we use the fact that the $\Unit$-, $\Sigma$- and $\Pi_\rep$- type formers have actions on homotopies between equivalences.
\end{proof}

We can then consider the equivalence relations $(\sim)$ obtained by taking the propositional truncation of $(\approx)$.
The above lemmas approximately say that $\alpha$ is the functorial action of a $1$-groupoid morphism from the setoid $(\Ob_{\Th_{\Ehat}},\sim)$ to the $1$-groupoid of objects of $\Th_{\Ehat}^\infty$ and homotopy classes of equivalences, and analogous statements for $\beta$ and $\beta^\rep$.

We will also need the following additional lemmas.
\begin{lem}\label{lem:ty_uniq}
  Given elements $p_1,p_2 : \Gamma_1 \approx \Gamma_2$ and $q : A_1 \approx_{p_1} A_2$ and a homotopy $h$ between $\alpha_{p_1}$ and $\alpha_{p_2}$, we can construct an element of $A_1 \approx_{p_2} A_2$.
\end{lem}
\begin{proof}
  By induction over the structure of $A_1$ and $A_2$.

  In the base case, $q : \iS(\sigma_1) \approx_{p_1} \iS(\sigma_2)$, \ie 
  \[ (\gamma_1:K(\Gamma_1),\gamma_2:K(\Gamma_2),\gamma_e:\alpha_{p_1}(\gamma_1,\gamma_2) \vdash q(\gamma_e) : K(\sigma_1)(\gamma_1) \simeq_{\partial \iS} K(\sigma_2)(\gamma_2)) \in \Th_{\Ehat}^\infty. \]
  By transporting over the homotopy between $\alpha_{p_1}$ and $\alpha_{p_2}$, we can construct 
  \[ (\gamma_1:K(\Gamma_1),\gamma_2:K(\Gamma_2),\gamma_e:\alpha_{p_2}(\gamma_1,\gamma_2) \vdash q'(\gamma_e) : K(\sigma_1)(\gamma_1) \simeq_{\partial \iS} K(\sigma_2)(\gamma_2)) \in \Th_{\Ehat}^\infty, \]
  which is an element of $\iS(\sigma_1) \approx_{p_2} \iS(\sigma_2)$.
 
  In other cases we use the fact that the $\Unit$-, $\Sigma$- and $\Pi_\rep$- type formers have actions on homotopies between equivalences.
\end{proof}

\begin{lem}\label{lem:ty_fibrancy}
  Given $p : \Gamma_1 \approx \Gamma_2$ and $(\Gamma_1 \vdash A_1\ \type) \in \Th_{\Ehat}$, we construct a type $A_2$ and an element $q : A_1 \approx_p A_2$.
\end{lem}
\begin{proof}
  By induction over the structure of $A_1$.

  In the base case, $A_1(\gamma_2) = \iS(\sigma_1(\gamma_1))$ where $(\Gamma_1 \vdash \sigma_1(\gamma_1) : \partial \iS)$.
  By transporting $K(\sigma_1)$ over $\alpha_p$, we can find $(\gamma_2 : K(\Gamma_2) \vdash \sigma_2(\gamma_2) : \partial \iS)$
  and
  \[ (\gamma_1:K(\Gamma_1),\gamma_2:K(\Gamma_2),\gamma_e:\alpha_p(\gamma_1,\gamma_2) \vdash \sigma_e(\gamma_e) : K(\sigma_1)(\gamma_1) \simeq_{\partial \iS} \sigma_2(\gamma_2). \]
  We then pose $A_2(\gamma_2) \triangleq \iS(K^{-1}(\sigma_2)(\gamma_2))$, and we have $\sigma_e : A_1 \approx_p A_2$.

  The recursive cases are straightforward.
\end{proof}

\begin{lem}\label{lem:approx_subst}
  Take two types $(\Delta_1 \vdash A_1\ \type), (\Delta_2 \vdash A_2\ \type) \in \Th_{\Ehat}$ and two morphisms $f_1 : \Gamma_1 \to \Delta_1$ and $f_2 : \Gamma_2 \to \Delta_2$.
  Given elements $p_\Gamma : \Gamma_1 \approx \Gamma_2$, $p_\Delta : \Delta_1 \approx \Delta_2$ and $q : A_1 \approx_{p_\Delta} A_2$ and an identification
  \[ (\gamma_1:K(\Gamma_1),\gamma_2:K(\Gamma_2),\gamma_e:\alpha_{p_\Gamma}(\gamma_1,\gamma_2) \vdash r(\gamma_e) : \alpha_{p_\Delta}(K(f_1)(\gamma_1),K(f_2)(\gamma_2))) \in \Th^\infty_{\Ehat}, \]
  we can construct an element $q[r] : A_1[f_1] \approx_{p_\Gamma} A_2[f_2]$ and a homotopy between $\alpha_{q[r]}$ and the equivalence $\alpha_q[r] : K(A_1[f]) \simeq K(A_2[f])$.
\end{lem}
\begin{proof}
  By induction over the structure of the types $A_1$ and $A_2$.

  In the base case, we have 
  \[ (\delta_1:K(\Delta_1),\delta_2:K(\Delta_2),\delta_e:\alpha_p(\delta_1,\delta_2) \vdash q(\delta_e) : K(\sigma_1)(\delta_1) \simeq_{\partial \iS} K(\sigma_2)(\delta_2)) \in \Th_{\Ehat}^\infty, \]
  and we can simply pose $q[r](\gamma_e) \triangleq q(r(\gamma_e))$. It is then clear that $\alpha_{q[r]}$ and $\alpha_q[r]$ are homotopic.

  The recursive cases use the fact that the type formers $\Sigma$ and $\Pi_\rep$ preserve homotopies between equivalences.
\end{proof}

We now show that these equivalence relations assemble into a fibrant congruence over $\Th_{\Ehat}$. 
We start by defining the equivalence relations, we will prove that the $(\reppre)$-CwF operations are preserved in a second step.
\begin{con}
  We define the data $\widetilde{\Th_{\Ehat}}$ of equivalence relations and dependent equivalence relations on the objects, morphisms, types and terms of the CwF $\Th_{\Ehat}$.
  \begin{itemize}
  \item Given $\Gamma_1,\Gamma_2 \in \Th_{\Ehat}$, we have $(\Gamma_1 \sim \Gamma_2) \in \widetilde{\Th_{\Ehat}}$ if there merely exists $p : \Gamma_1 \approx \Gamma_2$.
  \item Given $(\Gamma_1 \vdash A_1\ \type), (\Gamma_2 \vdash A_2\ \type) \in \Th_{\Ehat}$, we have $(A_1 \sim A_2) \in \widetilde{\Th_{\Ehat}}$ if there merely exists $p : \Gamma_1 \approx \Gamma_2$ and $q : A_1 \approx_p A_2$.
    The relation on representable types is defined similarly.
  \item Given $(\Gamma_1 \vdash a_1 : A_1), (\Gamma_2 \vdash a_2 : A_2) \in \Th_{\Ehat}$, we have $(a_1 \sim a_2) \in \widetilde{\Th_{\Ehat}}$ if there merely exists $p : \Gamma_1 \approx \Gamma_2$, $q : A_1 \approx_p A_2$ and an identification
    \[ (\gamma_1:K(\Gamma_1),\gamma_2:K(\Gamma_2),\gamma_e:\alpha_p(\gamma_1,\gamma_2) \vdash r(\gamma_e) : \beta_q(\gamma_e,K(a_1)(\gamma_1),K(a_2)(\gamma_2))) \in \Th^\infty_{\Ehat}. \]
  \item Given $f_1 : \Gamma_1 \to \Delta_1$ and $f_2 : \Gamma_2 \to \Delta_2$, we have $(f_1 \sim f_2) \in \widetilde{\Th_{\Ehat}}$ if there merely exists $p : \Gamma_1 \approx \Gamma_2$, $q : \Delta_1 \approx \Delta_2$ and an identification
    \[ (\gamma_1:K(\Gamma_1),\gamma_2:K(\Gamma_2),\gamma_e:\alpha_p(\gamma_1,\gamma_2) \vdash r(\gamma_e) : \alpha_q(K(f_1)(\gamma_1),K(f_2)(\gamma_2))) \in \Th^\infty_{\Ehat}. \]
  \end{itemize}

  The reflexivity, symmetry and transitivity of these relations follow from~\cref{lem:approx_refl} and~\cref{lem:approx_trans}.
\end{con}

\begin{lem}
  The dependent equivalence relations of $\widetilde{\Th_{\Ehat}}$ are fibrant.
\end{lem}
\begin{proof}
  The fibrancy of the relation on types follows from~\cref{lem:ty_fibrancy}.
  Fibrancy of the relations on terms and morphisms is proven by transport over the equivalences $\alpha_p$, $\beta_q$ and $\alpha_q$, along with the fact that the actions of $K$ on terms are bijective.
\end{proof}

\begin{lem}\label{lem:sim_uip_simpl_ty}
  Take $p : \Gamma_1 \approx \Gamma_2$ and types $(\Gamma_1 \vdash A_1\ \type), (\Gamma_2 \vdash A_2\ \type) \in \Th_{\Ehat}$.
  Then $(A_1 \sim A_2) \in \widetilde{\Th_{\Ehat}}$ if and only if there merely exists an element $q : A_1 \approx_p A_2$.
\end{lem}
\begin{proof}
  The reverse implication is evident. The forward implication follows from~\cref{lem:ty_uniq}.
\end{proof}

\begin{lem}\label{lem:sim_uip_simpl_tm}
  Take elements $p : \Gamma_1 \approx \Gamma_2$ and $q : A_1 \approx_p A_2$ and terms $(\Gamma_1 \vdash a_1 : A_1), (\Gamma_2 \vdash a_2 : A_2) \in \Th_{\Ehat}$.
  Then $(a_1 \sim a_2) \in \widetilde{\Th_{\Ehat}}$ if and only if there merely exists an identification
  \[ (\gamma_1:K(\Gamma_1),\gamma_2:K(\Gamma_2),\gamma_e:\alpha_p(\gamma_1,\gamma_2) \vdash r(\gamma_e) : \beta_q(K(a_1)(\gamma_1),K(a_2)(\gamma_2))) \in \Th^\infty_{\Ehat}. \]
\end{lem}
\begin{proof}
  The reverse implication is evident.
  For the forward implication, we can find $p_a : \Gamma_1 \approx \Gamma_2$ and $q_a : A_1 \approx_{p_a} A_2$ and an identification
  \[ (\gamma_1:K(\Gamma_1),\gamma_2:K(\Gamma_2),\gamma_e:\alpha_{p_a}(\gamma_1,\gamma_2) \vdash r_a(\gamma_e) : \beta_{q_a}(K(a_1)(\gamma_1),K(a_2)(\gamma_2))) \in \Th^\infty_{\Ehat}. \]

  By~\cref{lem:approx_uip}, we have homotopies between $\alpha_{p}$ and $\alpha_{p_a}$ and between $\beta_{q}$ and $\beta_{q_a}$.
  By transporting over these homotopies, we obtain
  \[ (\gamma_1:K(\Gamma_1),\gamma_2:K(\Gamma_2),\gamma_e:\alpha_{p}(\gamma_1,\gamma_2) \vdash r_a(\gamma_e) : \beta_{q}(K(a_1)(\gamma_1),K(a_2)(\gamma_2))) \in \Th^\infty_{\Ehat}, \]
  as needed.
\end{proof}

\begin{lem}\label{lem:sim_uip_simpl_hom}
  Take elements $p : \Gamma_1 \approx \Gamma_2$ and $q : \Delta_1 \approx_p \Delta_2$ and morphisms $f_1 : \Gamma_1 \to \Delta_1$ and $f_2 : \Gamma_2 \to \Delta_2$.
  Then $(f_1 \sim f_2) \in \widetilde{\Th_{\Ehat}}$ if and only if there merely exists an identification
  \[ (\gamma_1:K(\Gamma_1),\gamma_2:K(\Gamma_2),\gamma_e:\alpha_p(\gamma_1,\gamma_2) \vdash r(\gamma_e) : \alpha_q(K(f_1)(\gamma_1),K(f_2)(\gamma_2))) \in \Th^\infty_{\Ehat}. \]
\end{lem}
\begin{proof}
  Similar to~\cref{lem:sim_uip_simpl_tm}.
\end{proof}

It remains to prove that these relations are preserved by all $(\reppre)$-CwF operations.
\begin{lem}
  The relations of $\widetilde{\Th_{\Ehat}}$ form a fibrant congruence over the $(\reppre)$-CwF $\Th_{\Ehat}$.
\end{lem}
\begin{proof}
  We prove that all operations are preserved.
  The proof is quite lengthy due to the large number of operations in the theory of $(\reppre)$-CwFs.
  However, all cases are rather straightforward.
  The main idea is to use~\cref{lem:sim_uip_simpl_ty}, \cref{lem:sim_uip_simpl_tm} and~\cref{lem:sim_uip_simpl_hom} to ensure that we do not obtain different elements of $\Gamma_1 \approx \Gamma_2$ or $A_1 \approx_p A_2$.
  \Cref{lem:approx_subst} is also needed in the rules that involve substitution.
  
  \begin{description}
  \item[Identity morphism]
    Take $p : \Gamma_1 \approx \Gamma_2$.
    Then over the context $(\gamma_1:K(\Gamma_1),\gamma_2:K(\Gamma_2),\gamma_e:\alpha_p(\gamma_1,\gamma_2))$, the identification
    \[ (\gamma_e : \alpha_p(K(\id)(\gamma_1),K(\id)(\gamma_2))) \]
    witnesses the fact that $\id_{\Gamma_1} \sim \id_{\Gamma_2}$.

  \item[Composition of morphisms]
    Take morphisms $f_1 : \Gamma_1 \to \Delta_1$, $g_1 : \Delta_1 \to \Theta_1$, $f_2 : \Gamma_2 \to \Delta_2$ and $g_2 : \Delta_2 \to \Theta_2$ such that
    $f_1 \sim f_2$ and $g_1 \sim g_2$.
    Relying on~\cref{lem:sim_uip_simpl_hom}, we can find elements
    $p_f : \Gamma_1 \approx \Gamma_2$, $q_f : \Delta_1 \approx \Delta_2$ and $q_g : \Theta_1 \approx \Theta_2$ and identifications
    \[ (\gamma_1:K(\Gamma_1),\gamma_2:K(\Gamma_2),\gamma_e:\alpha_{p_f}(\gamma_1,\gamma_2) \vdash r_f(\gamma_e) : \alpha_{q_f}(K(f_1)(\gamma_1),K(f_2)(\gamma_2))) \in \Th^\infty_{\Ehat} \]
    and 
    \[ (\delta_1:K(\Delta_1),\delta_2:K(\Delta_2),\delta_e:\alpha_{q_f}(\delta_1,\delta_2) \vdash r_g(\delta_e) : \alpha_{q_g}(K(g_1)(\delta_1),K(g_2)(\delta_2))) \in \Th^\infty_{\Ehat}. \]

    We then have, over the context $(\gamma_1:K(\Gamma_1),\gamma_2:K(\Gamma_2),\gamma_e:\alpha_{p_f}(\gamma_1,\gamma_2))$, an identification
    \[ (r_g(r_f(\gamma_e)) : \alpha_{q_g}(K(g_1 \circ f_1)(\gamma_1)), K(g_2 \circ f_2)(\gamma_2)) \in \Th^\infty_{\Ehat}, \]
    witnessing that $(g_1 \circ f_1) \sim (g_2 \circ g_2)$.

  \item[Functorial action on terms]
    Similar to composition of morphisms, using~\cref{lem:sim_uip_simpl_tm} instead.
    
  \item[Functorial action on types]
    Take morphisms $f_1 : \Gamma_1 \to \Delta_1$ and $f_2 : \Gamma_2 \to \Delta_2$ and types $(\Delta_1 \vdash A_1\ \type)$ and $(\Delta_2 \vdash A_2\ \type)$ such that $f_1 \sim f_2$ and $A_1 \sim A_2$.
    
    Using~\cref{lem:sim_uip_simpl_ty}, we can find elements $p_f : \Gamma_1 \approx \Gamma_2$, $q_f : \Delta_1 \approx \Delta_2$ and $q_A : A_1 \approx_{q_f} A_2$ and an identification
    \[ (\gamma_1:K(\Gamma_1),\gamma_2:K(\Gamma_2),\gamma_e:\alpha_{p_f}(\gamma_1,\gamma_2) \vdash r_f(\gamma_e) : \alpha_{q_f}(K(f_1)(\gamma_1),K(f_2)(\gamma_2))) \in \Th^\infty_{\Ehat}. \]

    We then conclude using~\cref{lem:approx_subst}.
    
  \item[Functorial action on representable types]
    Similar to the functorial action on types.
    
  \item[Context extension]
    Let $(\Gamma_1 \vdash A_1\ \type), (\Gamma_2 \vdash A_2\ \type) \in \Th_{\Ehat}$ be two types such that $A_1 \sim A_2$.

    We can find $p : \Gamma_1 \approx \Gamma_2$ and $q : A_1 \approx_p A_2$.
    Then $(p,q) : (\Gamma_1.A_1) \approx (\Gamma_2.A_2)$.
  \item[Substitution extension]
    Let $(\Delta_1 \vdash A_1\ \type), (\Delta_2 \vdash A_2\ \type) \in \Th_{\Ehat}$ be two types such that $A_1 \sim A_2$, let
    $f_1 : \Gamma_1 \to \Delta_1$ and $f_2 : \Gamma_2 \to \Delta_2$ be two morphisms such that $f_1 \sim f_2$, and let
    $(\Gamma_1 \vdash a_1 : A_1[f_1])$ and $(\Gamma_2 \vdash a_2 : A_2[f_2])$ be two terms such that $a_1 \sim a_2$.

    We can find elements $p_A : \Delta_1 \approx \Delta_2$, $q_A : A_1 \approx_{p_A} A_2$, $p_f : \Gamma_1 \approx \Gamma_2$, $q_f : \Delta_1 \approx \Delta_2$ along with identifications
    \[ (\gamma_1:K(\Gamma_1),\gamma_2:K(\Gamma_2),\gamma_e:\alpha_{p_f}(\gamma_1,\gamma_2) \vdash r_f(\gamma_e) : \alpha_{q_f}(K(f_1)(\gamma_1),K(f_2)(\gamma_2))) \in \Th^\infty_{\Ehat}. \]
    By~\cref{lem:approx_subst}, we have $q_A[r_f] : A_1[f_1] \approx_{p_f} A_2[f_2]$ and a homotopy between $\beta_{q_A[r_f]}$ and $\beta_{q_A}[r_f]$.

    Using~\cref{lem:sim_uip_simpl_tm}, we can then find
    \[ (\gamma_1:K(\Gamma_1),\gamma_2:K(\Gamma_2),\gamma_e:\alpha_{p_f}(\gamma_1,\gamma_2) \vdash r_a(\gamma_e) : \beta_{q_A[r_f]}(K(a_1)(\gamma_1),K(a_2)(\gamma_2))) \in \Th^\infty_{\Ehat} \]

    The homotopy between $\beta_{q_A[r_f]}$ and $\beta_{q_A}[r_f]$ has an underlying transport function
    \[ (\gamma_1,\gamma_2,\gamma_e:\alpha_{p_f}(\gamma_1,\gamma_2),a_1,a_2,a_e:\beta_{q_A[r_f]}(\gamma_e,a_1,a_2) \vdash h(\gamma_e,a_e) : q_A(r_f(\gamma_e), a_1,a_2)). \]

    Then the identification
    \[ ((r_f(\gamma_e),h(\gamma_e,r_a(\gamma_e))) : \alpha_{p_A,q_A}(K(\angles{f_1,a_1})(\gamma_1),K(\angles{f_2,a_2})(\gamma_2))) \in \Th^\infty_{\Ehat} \]
    over the context $(\gamma_1:K(\Gamma_1),\gamma_2:K(\Gamma_2),\gamma_e:\alpha_{p_f}(\gamma_1,\gamma_2))$
    witnesses that $\angles{f_1,a_1} \sim \angles{f_2,a_2}$.

  \item[Substitution projections]
    Let $f_1 : \Gamma_1 \to \Delta_1.A_1$ and $f_2 : \Gamma_2 \to \Delta_2.A_2$ be two morphisms such that
    $f_1 \sim f_2$ and $A_1 \sim A_2$.

    Using~\cref{lem:sim_uip_simpl_hom}, we can find elements $p_f : \Delta_1 \approx \Delta_2$, $p_A : \Gamma_1 \approx \Gamma_2$, $q_A : A_1 \approx_{p_A} A_2$ and an identification
    \[ (\gamma_1:K(\Gamma_1),\gamma_2:K(\Gamma_2),\gamma_e:\alpha_{p_f}(\gamma_1,\gamma_2) \vdash r_f(\gamma_e) : \alpha_{(p_A,q_A)}(K(f_1)(\gamma_1),K(f_2)(\gamma_2))) \in \Th^\infty_{\Ehat}. \]
    By definition of $\alpha_{(p_A,q_A)}$, we can decompose $r_f$ into identifications
    \[ (\gamma_1,\gamma_2,\gamma_e:\alpha_{p_f}(\gamma_1,\gamma_2) \vdash \pi_1(r_f)(\gamma_e) : \alpha_{p_A}(K(\pi_1(f_1))(\gamma_1)),K(\pi_1(f_2))(\gamma_2))) \in \Th^\infty_{\Ehat} \]
    and
    \[ (\gamma_1,\gamma_2,\gamma_e:\alpha_{p_f}(\gamma_1,\gamma_2) \vdash \pi_2(r_f)(\gamma_e) : \beta_{q_A}(\pi_1(r_f)(\gamma_e),K(\pi_2(f_1))(\gamma_1)),K(\pi_2(f_2))(\gamma_2))) \in \Th^\infty_{\Ehat}. \]

    Then $\pi_1(r_f)$ witnesses the fact that $\pi_1(f_1) \sim \pi_1(f_2)$.

    By~\cref{lem:approx_subst}, we have an element $q_A[r_f] : A_1[\pi_1(r_f)] \approx_{p_f} A_2[\pi_2(r_f)]$ and a homotopy between $\beta_{q_A[r_f]}$ and $\beta_{q_A}[r_f]$.
    Up to transport over this homotopy, $\pi_2(r_f)$ witnesses the fact that $\pi_2(f_1) \sim \pi_2(f_2)$.

  \item[$\Unit$-type structure]
    Straightforward.
    Note that we also need to consider the $\Unit$-type structure on representable types.

  \item[$\Sigma$-type structure] \hfill \\
    Similar to context extensions, substitution extension and substitution projections.
    We also need to consider the $\Sigma$-type structure on representable types.

  \item[$\Pi_\rep$-type structure] \hfill \\
    \begin{description}
    \item[Type former]
      Let $(\Gamma_1.A_1 \vdash B_1\ \type), (\Gamma_2.A_2 \vdash B_2\ \type)$ be two dependent types such that $A_1 \sim A_2$ and $B_1 \sim B_2$.
      Using~\cref{lem:sim_uip_simpl_ty}, we can find $p_A : \Gamma_1 \approx \Gamma_2$, $q_A : A_1 \approx^\rep_{p_A} A_2$ and $q_B : B_1 \approx_{(p_A,q_A)} B_2$.
      Then $(q_A,q_B) : \Pi_\rep(A_1,B_1) \approx_{p_A} \Pi_\rep(A_2,B_2)$, as needed.
    \item[Application]
      Let $(\Gamma_1 \vdash f_1 : \Pi_\rep(A_1,B_1)), (\Gamma_2 \vdash f_2 : \Pi_\rep(A_2,B_2))$ be two terms such that $f_1 \sim f_2$, and
      $(\Gamma_1 \vdash a_1 : A_1), (\Gamma_2 \vdash a_2 : A_2)$ be two terms such that $a_1 \sim a_2$.
      We have $p_A$, $q_A$ and $q_B$ as before.

      Using~\cref{lem:sim_uip_simpl_tm}, we can find identifications
      \[ (\gamma_1,\gamma_2,\gamma_e:\alpha_{p_A}(\gamma_1,\gamma_2) \vdash r_f(\gamma_e) : \beta_{(q_A,q_B)}(\gamma_e,K(f_1)(\gamma_1),K(f_2)(\gamma_2))) \in \Th^\infty_{\Ehat} \]
      and
      \[ (\gamma_1,\gamma_2,\gamma_e:\alpha_{p_A}(\gamma_1,\gamma_2) \vdash r_a(\gamma_e) : \beta_{q_A}(\gamma_e,K(a_1)(\gamma_1),K(a_2)(\gamma_2))) \in \Th^\infty_{\Ehat}. \]

      By definition of $\beta_{(q_A,q_B)}$, we can apply $r_f(\gamma_e)$ to $r_a(\gamma_e)$ to obtain an identification in
      \[ \beta_{q_B}(\gamma_e, K(\app(f_1,a_1))(\gamma_1), K(\app(f_2,a_2))(\gamma_2)) \]
      witnessing that $\app(f_1,a_1) \sim \app(f_2,a_2)$.
    \item[Lambda]
       Let $(\Gamma_1.A_1 \vdash b_1 : B_1), (\Gamma_2.A_1 \vdash b_2 : B_2)$ be two terms such that $b_1 \sim b_2$.
       We have $p_A$, $q_A$ and $q_B$ as before.
       By~\cref{lem:sim_uip_simpl_tm}, we can find an identification
       \[ (r_b(\gamma_e,a_e) : \beta_{q_B}((\gamma_e,a_e),K(b_1)(\gamma_1,a_1),K(b_2)(\gamma_2,a_2))) \in \Th^\infty_{\Ehat} \]
       over the context $(\gamma_1,\gamma_2,\gamma_e:\alpha_{p_A}(\gamma_1,\gamma_2),a_1,a_2,a_e:\beta_{q_A}(a_1,a_2))$.

       By definition of $\beta_{(q_A,q_B)}$, we can construct an element
       \begin{alignat*}{1}
         & \lam(\lambda (a_1 : A_1)\ ((a_2,a_e) := \overrightarrow{\beta_{q_A}}(a_1)) \mapsto r_b(\gamma_e,a_e)) \\
         & \quad : \beta_{(q_A,q_B)}(\gamma_e,K(\lam(b_1))(\gamma_1),K(\lam(b_2)(\gamma_2))) 
       \end{alignat*}
       over the context $(\gamma_1,\gamma_2,\gamma_e:\alpha_{p_A}(\gamma_1,\gamma_2))$.
       Thus $\lam(b_1) \sim \lam(b_2)$, as needed.
       \qedhere{}
    \end{description}
  \end{description}
\end{proof}

\begin{lem}\label{lem:sim_subseteq_kernel_helper}
  Given $p : \Gamma_1 \approx \Gamma_2$ we have $F(\Gamma_1) = F(\Gamma_2)$ and $G(\alpha_p)$ is (homotopic to) the identity equivalence in $\Th_E^1$.

  Given $q : A_1 \approx_p A_2$ we have $F(A_1) = F(A_2)$ and $G(\beta_p)$ is (homotopic to) the identity equivalence in $\Th_E^1$.
\end{lem}
\begin{proof}
  By mutual induction on the structure of contexts and terms.

  The recursive cases follow from the fact that $F$ is a $(\reppre)$-CwF morphism, that $G$ is a $(\repinfty)$-morphism and that the constructions of $\alpha_p$ and $\beta_q$ preserve identity equivalences.

  In the base case, we have $q : \iS(\sigma_1) \approx_p \iS(\sigma_2)$, \ie
  \[ (\gamma_1:K(\Gamma_1),\gamma_2:K(\Gamma_2),\gamma_e:\alpha_p(\gamma_1,\gamma_2) \vdash q(\gamma_e) : K(\sigma_1)(\gamma_1) \simeq_{\partial \iS} K(\sigma_2)(\gamma_2)) \in \Th_{\Ehat}^\infty. \]
  By applying $G$, we obtain
  \[ (\gamma_1,\gamma_2,\gamma_e:G(\alpha_p)(\gamma_1,\gamma_2) \vdash G(q)(\gamma_e) : G(K(\sigma_1))(\gamma_1) \simeq_{\partial \iS} G(K(\sigma_2))(\gamma_2)) \in \Th_{E}^1. \]
  By the induction hypothesis, $F(\Gamma_1) = F(\Gamma_2)$ and $\alpha_p$ is the identity equivalence.
  We also recall that $\Th_E^1$ satisfies equality reflection.
  This implies that $G(K(\sigma_1)) = G(K(\sigma_2))$, or equivalently $L(F(\sigma_1)) = L(F(\sigma_2))$.
  Since the actions of $L$ on terms are bijective, we have $F(\sigma_1) = F(\sigma_2)$.
  Furthermore, since $G(\alpha_q)$ is an equivalence defined by transport over an equality, it is the identity equivalence, as needed.
\end{proof}

\begin{lem}\label{lem:sim_subseteq_kernel}
  We have $\widetilde{\Th_{\Ehat}} \subseteq \ker_F$.
\end{lem}
\begin{proof}
  We prove the inclusions of the relations for each sort.
  \begin{itemize}
  \item Take $(\Gamma_1 \sim \Gamma_2) \in \widetilde{\Th_{\Ehat}}$.
    We have some $p : \Gamma_1 \approx \Gamma_2$. By~\cref{lem:sim_subseteq_kernel_helper}, $F(\Gamma_1) = F(\Gamma_2)$, as needed.
  \item Take $(A_1 \sim A_2) \in \widetilde{\Th_{\Ehat}}$.
    We have some $q : A_1 \approx_p A_2$. By~\cref{lem:sim_subseteq_kernel_helper}, $F(A_1) = F(A_2)$, as needed.
  \item Take $(a_1 \sim a_2) \in \widetilde{\Th_{\Ehat}}$.
    We have some $q : A_1 \approx_p A_2$ and an identification
    \[ (\gamma_1:K(\Gamma_1),\gamma_2:K(\Gamma_2),\gamma_e:\alpha_p(\gamma_1,\gamma_2) \vdash r(\gamma_e) : \beta_q(K(a_1)(\gamma_1),K(a_2)(\gamma_2))) \in \Th^\infty_{\Ehat}. \]
    By~\cref{lem:sim_subseteq_kernel_helper}, we know that $F(A_1) = F(A_2)$ and that $G(\beta_q)$ is the identity equivalence in $\Th_E^1$.
    Therefore, $G(r)$ proves that $G(K(a_1)) = G(K(a_2))$ in $\Th_E^1$, or equivalently $L(F(a_1)) = L(F(a_2))$.
    Since the actions of $L$ on terms are bijective, we have $F(a_1) = F(a_2)$, as needed.

  \item The case of morphisms is similar to the case of terms.
    \qedhere
  \end{itemize}
\end{proof}

\begin{lem}\label{lem:sim_marked_hpty}
  For every marked homotopy
  \[ (\gamma:\Gamma \vdash h(\gamma) : x(\gamma) \sim_{\iS(\sigma)} y(\gamma)) \in E, \]
  we have $(x \sim y) \in \widetilde{\Th_{\Ehat}}$
\end{lem}
\begin{proof}
  Take a marked homotopy $h$ as in the statement.
  
  Let $p : \Gamma \approx \Gamma$ and $q : \iS(\sigma) \approx_p \iS(\sigma)$ be $r_\Gamma$ and $r_{\iS(\sigma)}$ from~\cref{lem:approx_refl}.

  It then suffices to construct an element of type 
  \[ \alpha_q(K(x)(\gamma_1), K(y)(\gamma_2)) \]
  over the context $(\gamma_1:K(\Gamma),\gamma_2:K(\Gamma),\alpha_p(\gamma_1,\gamma_2))$.

  Since $\alpha_p$ and $\alpha_q$ are homotopic to identity equivalences, it suffices to construct an element of type
  \[ K(x)(\gamma) \simeq K(y)(\gamma) \]
  over the context $(\gamma:K(\Gamma))$.

  We can then pick the element $\widehat{h}(\gamma) : K(x)(\gamma) \simeq K(y)(\gamma)$.
\end{proof}

\subsection{Proof of the main theorem}

We can now prove~\cref{thm:main_theorem}.
\begin{proof}[Proof of~\cref{thm:main_theorem}]
  We have constructed a fibrant congruence $\widetilde{\Th_{\Ehat}}$ over the $(\reppre)$-CwF $\Th_{\Ehat}$ and we can consider the quotient $\Quot(\widetilde{\Th_{\Ehat}})$ in $(\reppre)$-CwF.

  By~\cref{cor:fib_cong_tfib}, the quotient inclusion $\quot : \Th_{\Ehat} \to \Quot(\widetilde{\Th_{\Ehat}})$ is a non-split trivial fibration.

  By~\cref{lem:sim_subseteq_kernel} and the universal property of the quotient, the map $F : \Th_{\Ehat} \to \Th_E$ factors through the quotient inclusion.
  We obtain a map $R : \Quot(\widetilde{\Th_{\Ehat}}) \to \Th_E$.

  By~\cref{lem:sim_marked_hpty} and the universal property of $\Th_E$, we obtain a section $S : \Th_E \to \Quot(\widetilde{\Th_{\Ehat}})$.
  
  \[ \begin{tikzcd}
      &
      \Quot(\widetilde{\Th_{\Ehat}})
      \ar[rd, "R"']
      & \\
      \Th_{\Ehat}
      \ar[rr, "F"']
      \ar[ru, "\quot"]
      &&
      \Th_E
      \ar[lu, bend right, "S"']
    \end{tikzcd} \]

  The map $F$ is then a section of the quotient inclusion $\quot$.
  Since non-split trivial fibrations are closed under retracts, the map $F$ is a non-split trivial fibration.
\end{proof}

Our main theorem directly follows.
\begin{proof}[Proof of~\cref{thm:main_theorem_cor}]
  By~\cref{lem:weq_ehat} the map $\Th \to \Th_{\Ehat}$ is a weak equivalence in $\CMod_\Th$.
  By~\cref{thm:main_theorem} the map $\Th \to \Th_E$ is a non-split trivial fibration in $\CCwf_\reppre$, hence a non-split weak equivalence in $\CMod_\Th$.
  Their composite $\Th \to \Th_E$ is then also a non-split weak equivalence in $\CMod_\Th$.
\end{proof}


\section{Future directions}\label{sec:discussion}

\subsection{Alternative proofs of theorem~8.1}

The proof of~\cref{thm:main_theorem} is very close to Hofmann's original proof of the conservativity of extensional type theory over intensional type theory.
We discuss potential ways to obtain simpler proofs of~\cref{thm:main_theorem}.

\subsubsection*{Local universes}

In presence of enough universes in $\Th_\Ehat^\infty$, there would be a simpler way to define the relations of the congruence $\widetilde{\Th_\Ehat}$.
Indeed, we could say that $(\Gamma \sim \Delta) \in \widetilde{\Th_\Ehat}$ when there merely exists $p : \mathsf{code}(K(\Gamma)) \simeq_\CU \mathsf{code}(K(\Delta))$, where $\mathsf{code}(-)$ gives the code of a context or type in the universe.

In the absence of universe, we can try to use local universes.
A local universe in $\Th_\Ehat^\infty$ is a type $(V \vdash E\ \type) \in \Th_\Ehat^\infty$.
By induction on the shape of contexts and types of $\Th_\Ehat$, we would define for every $\Gamma \in \Th_\Ehat$ a local universe $(V_\Gamma,E_\Gamma)$ and an element $(\chi_\Gamma : V_\Gamma) \in \Th_\Ehat^\infty$ such that $K(\Gamma) = E_\Gamma[\chi_\Gamma]$.
We could then say that $(\Gamma \sim \Delta) \in \widetilde{\Th_\Ehat}$ when $(V_\Gamma,E_\Gamma) = (V_\Delta,E_\Delta)$ and there merely exists $p : \chi_\Gamma \simeq_{V_\Gamma} \chi_\Delta$.

Unfortunately, it is not possible to define these local universes: we would need arbitrary $\Pi$-types in $\Th_\Ehat^\infty$, but we only have $\Pi_\rep$-types.
If we could conservatively embed $\Th_\Ehat^\infty$ into a $(\lccinfty)$-CwF, this method would be applicable.

\subsubsection*{Higher-order congruences}

Another possible simplification would be obtained by working internally to $\CPsh(\Th_{\Ehat})$ instead of externally.

Consider the following external diagram:
\[
  \begin{tikzcd}
    \Th
    \ar[r]
    \ar[d]
    & \Th_E
    \ar[d, equal]
    \\
    \Th_{\Ehat}
    \ar[r, "F"]
    \ar[d, "K"]
    & \Th_E
    \ar[d, "L"]
    \\
    \Th_{\Ehat}^{\infty}
    \ar[r, "G"]
    &
    \Th_E^1\rlap{\ .}
  \end{tikzcd}
\]

Following ideas from~\cite{InternalSconing}, it is possible to internalize this diagram internally to $\CPsh(\Th_{\Ehat})$ as follows:
\[
  \begin{tikzcd}
    \Th
    \ar[r]
    \ar[d]
    & \Th_E
    \ar[d]
    \\
    \CTele_{\MT_{\Ehat}}
    \ar[r]
    \ar[d]
    & \CI_E
    \ar[d]
    \\
    \CI_{\Ehat}^{\infty}
    \ar[r]
    &
    \CI_E^1\rlap{\ .}
  \end{tikzcd}
\]
Here $\CTele_{\MT_{\Ehat}}$ is an internal $(\reppre)$-CwF called the \emph{telescopic contextualization} of the family $\MT_{\Ehat}$; its contexts are the telescopes of the family $\MT_{\Ehat}$.
Similarly, $\CI_{\Ehat}^\infty$, $\CI_E$ and $\CI_E^1$ are internalization of $\Th_{\Ehat}$, $\Th_E$ and $\Th_E^1$.
Using the restriction operation from~\cite{InternalSconing}, they would be defined as $\CI_{\Ehat}^\infty \triangleq K^\ast(\CTele_{\MT_\Ehat^\infty})$,  $\CI_{E} \triangleq F^\ast(\CTele_{\MT_E})$ and  $\CI_{\Ehat}^\infty \triangleq K^\ast(G^\ast(\CTele_{\MT_E^1}))$.
The internal versions of $\Th$ and $\Th_E$ should be constant over , \ie defined by restriction of their external counterparts over the unique functor $\Th_{\Ehat} \to 1_\CCat$.

It should then be possible to define a ``higher-order congruence'' over $\MT_{\Ehat}$.
The advantage of this approach is that higher-order congruences only consist of equivalence relations on types and terms, not on contexts and substitutions;
we would avoid the parts of the current proof that have to deal with contexts and substitutions.

This higher-order congruence should induce a fibrant ``first-order'' congruence over the $(\reppre)$-CwF $\CTele_{\MT_\Ehat}$.
We can then consider the internal quotient of this fibrant congruence; it fits in the following diagram.
\[
  \begin{tikzcd}
    \Th
    \ar[rr]
    \ar[dd]
    && \Th_E
    \ar[ld]
    \ar[dd]
    \\
    & Q
    \ar[rd]
    &
    \\
    \CTele_{\MT_{\Ehat}}
    \ar[ru]
    \ar[rr]
    && \CI_E\rlap{\ .}
  \end{tikzcd}
\]

By externalizing the above diagram, \ie restricting it along the functor $1_{\Th_{\Ehat}} : 1_\CCat \to \Th_{\Ehat}$, we obtain the following external diagram:
\[
  \begin{tikzcd}
    \Th
    \ar[rr]
    \ar[dd]
    && \Th_E
    \ar[ld]
    \ar[dd, equal]
    \\
    & 1_{\Th_\Ehat}^\ast(Q)
    \ar[rd]
    &
    \\
    \Th_{\Ehat}
    \ar[ru]
    \ar[rr]
    && \Th_E\rlap{\ .}
  \end{tikzcd}
\]

We would then conclude as before, by observing that $\Th_{\Ehat} \to \Th_E$ becomes a retract of the quotient inclusion $\Th_{\Ehat} \to 1_{\Th_\Ehat}^\ast(Q)$.

This proof strategy requires some additional tools for axiomatizing models, morphisms of models and their properties in the internal language of presheaf toposes.

\subsection{Conjectures}

This paper includes several conjectures (\cref{clm:ff_sign}, \cref{clm:ext_univ_model_structure}, \cref{clm:characterization_morita_eqv}, \cref{conj:contractible_extension}, \cref{conj:embedding_from_ext_univ}).

The first three conjectures (\cref{clm:ff_sign}, \cref{clm:ext_univ_model_structure} and \cref{clm:characterization_morita_eqv}) should be mostly unproblematic, at least in a classical setting; The author has already drafts of their proofs.
It may however be rather difficult to obtain constructive proofs.

Results analogous to~\cref{clm:ff_sign} are well-known for algebraic theories and essentially algebraic theories, and it is possible to extend these results to first-order GATs.
For SOGATs, it seems that the simplest approach is a reduction to first-order GATs.
Indeed, the underlying $\lexpre$-CwF of a SOGAT $\Th$ is a first-order GAT that classifies the contextual models of $\Th$ (see~\cite[Theorem~7.30]{GeneralFrameworkTT}).

Proving~\cref{clm:ext_univ_model_structure} is mainly a matter of translating between the weakly stable identity types of $\Th$ and properties of the weak factorization systems of $\CMod_\Th^\cxl$ with regard to the fully faithful functor $\Init_\Th[-] : \Th^\op \to \CMod_\Th^\cxl$.
The theory $\Th$ being finitary explains why considering the finitely generated models of $\Th$ suffices.

The proof of~\cref{clm:characterization_morita_eqv} has already been sketched in~\cref{prf:characterization_morita_eqv}.

\Cref{conj:contractible_extension} should follow from strictification theorems for weakly stable identity types.
An approach for solving these coherence problems is discussed in~\cref{ssec:coherence_problems}.

\Cref{conj:embedding_from_ext_univ} should follow from gluing arguments that are very similar to the gluing arguments used in the proofs of~\cref{prop:the_the1_embedding}.
The main issue is that one likely has to use the internal language of some $\infty$-topos, which is not easily available in our setting.
A possible solution is also discussed in~\cref{ssec:coherence_problems}.

\subsection{Coherence problems and the Yoneda embedding}\label{ssec:coherence_problems}

Coherence problems arise from the use of $(\lexinfty)$- or $(\repinfty)$- CwFs to morally represent structured $\infty$-categories.
In particular, the notion of external univalence is expressed using weakly stable identity types.
Thus, in order to rely on external univalence, the two choices are either to only use weakly stable identity types, or to prove coherence theorem that allow for the strictification of the weakly stable identity types.

The Yoneda embedding is one of the most important constructions of category theory.
For applications, it seems very important to access the $\infty$-categorical Yoneda embedding at the level of $(\repinfty)$-CwFs.
If $\CC$ is a cofibrant $(\repinfty)$-CwF, it should be possible to find a model $\CD$ of HoTT (corresponding to $\infty$-presheaves over $\CC$) and a $(\repinfty)$-CwF morphism $\CC \to \CD$ that is essentially surjective on terms (corresponding to the fully faithful Yoneda embedding).

We hope that both of these issues will be solved in the future using cubical methods, thanks to the fact that the cubical model of HoTT is constructive (using a suitable cube category).
If $\CC$ is any $1$-category with a terminal object, then the category $[\CC^\op, \mathbf{cSet}]$ of cubical presheaves is equipped with a model of HoTT.
Indeed, we can construct the cubical set model internally to $\CPsh(\CC)$.
This internal construction corresponds externally to a functor $\CC \to \CMod_{\mathsf{HoTT}}^\op$.
By composing this functor with $1_\CC : 1_\CCat^\op \to \CC^\op$, we obtain an external model of HoTT; it can be shown that its underlying category is the category $[\CC^\op, \mathbf{cSet}]$ of cubical presheaves over $\CC$.
When $\CC$ is a $(\lexinfty)$-CwF or a $(\reppre)$-CwF, we then want to compare $\CC$ with this model $[\CC^\op, \mathbf{cSet}]$.
It should then suffice to compare, internally to $\CPsh(\CC)$, the CwF $\SSet_{\mathsf{rf}}$ of sets and Reedy-fibrant types with the cubical set model $\mathsf{cSet}$.
This problem seems tractable (using, for example, ideas from two-level type theory (2LTT) and the development of Reedy-fibrant diagrams in 2LTT~\cite{2LTT}).

\subsection{Proving $0$-truncatedness from normalization}

In order to apply our methods, we would need to be able to prove the (relative) $0$-truncatedness of $\Th_\Ehat^\infty$ in concrete examples.
We believe that such results can be obtained by lifting normalization proofs from $\Th_E$ to $\Th_{\Ehat}^\infty$.
These normalization results should even provide direct proofs of the fact that $\Th_{\Ehat} \to \Th_E$ is a \emph{split} trivial fibration, without going through any quotient.

Strict normalization for $\Th_E$ states that every term of $\Th_E$ admits a unique normal form.
Normalization is usually proven for the syntax of type theory, that is for the initial model $\Init_{\Th_E}$.
Proving normalization for $\Th_E$ instead is akin to proving normalization for all freely generated models of $\Th_E$ at once.
For most theories, proving normalization for $\Th_E$ is not more complicated than proving normalization for $\Init_{\Th_E}$.

Homotopy normalization for $\Th_{\Ehat}^\infty$ should then state that every term of $\Th_{\Ehat}^\infty$ admits a contractible space of normal forms (over contexts in the image of $\Th \to \Th_\Ehat^\infty$).
Stated appropriately, this is sufficient to obtain the $0$-truncatedness of $\Th_\Ehat^\infty$.
Proving homotopy normalization for $\Th_{\Ehat}^\infty$ is however much more complicated than proving strict normalization.
Indeed, where strict normalization proofs for type theory make use of set-valued logical relations, homotopy normalization proofs should use space-valued logical relations.

For this we essentially need to use $\infty$-presheaves over some $\infty$-categories, and the interpretation of HoTT into $\infty$-toposes.
But because we have used $(\lexinfty)$-CwFs or $(\repinfty)$-CwFs instead of other models of $\infty$-categories, we cannot directly use the known interpretation of HoTT into $\infty$-toposes.
Hopefully, the cubical methods mentioned in~\cref{ssec:coherence_problems} can be used instead.

We note that Uemura has recently written a proof of normalization and $0$-truncatedness for the initial models of $\infty$-type theories~\cite{UemuraNormalization}.
Only a small $\infty$-type theory with only $\Pi$-types has been considered so far, but the methods should also work for other $\infty$-type theories.
It does not seem realistic to directly use these results for our purposes (although it would be possible in principle, given a precise comparison between $(\repinfty)$-CwFs and $\infty$-type theories (representable map $\infty$-categories)).
But the fact that $0$-truncatedness results can be obtained for $\infty$-type theories indicates that similar methods should also yield the $0$-truncatedness results we need.


\section*{Acknowledgements}
The author would like to thank Thorsten Altenkirch, Martin Bidlingmaier, Paolo Capriotti, Thierry Coquand, Simon Huber, Ambrus Kaposi, András Kovács, Nicolai Kraus, Chaitanya Leena Subramaniam, Christian Sattler and Bas Spitters for expressing interest in this work and helpful discussions.

\bibliographystyle{alpha}
\bibliography{main}

\end{document}